\documentclass[twocolumn,preprintnumbers,superscriptaddress,amsmath,amssymb]{revtex4}
\usepackage{float}
\usepackage{graphicx}% Include figure files
\usepackage{dcolumn}% Align table columns on decimal point
\usepackage{bm}% bold math
\usepackage{epstopdf}
\usepackage{xcolor}
\usepackage{textcomp}
\usepackage{soul}% \st
\usepackage{csquotes}
\usepackage{amsbsy}
\usepackage{mathrsfs} % script-like, curvy letters.
\usepackage{amsmath}
\usepackage{bigints} % Big Integration
\usepackage{amsthm,amssymb} % Theorem Proof
\usepackage[utf8]{inputenc}
\usepackage[english]{babel}
\usepackage[shortlabels]{enumitem}
\usepackage{float}
\usepackage{mathrsfs,bigints,mathtools,dsfont}
\usepackage[colorlinks=true,linkcolor=blue,citecolor=blue]{hyperref}%
\usepackage[toc]{appendix}
\newcommand\norm[1]{\left\lVert#1\right\rVert}

\newtheorem{theorem}{Theorem}

\newtheorem{lemma}[theorem]{Lemma}

\newtheorem{corollary}[theorem]{Corollary}

\newtheorem*{remark}{Remark}

\begin{document}
	
	\title{{Synchronization in temporal simplicial complexes}}
	\author{Md Sayeed Anwar}\affiliation{Physics and Applied Mathematics Unit, Indian Statistical Institute, 203 B. T. Road, Kolkata 700108, India}
	\author{Dibakar Ghosh}\email{dibakar@isical.ac.in}\affiliation{Physics and Applied Mathematics Unit, Indian Statistical Institute, 203 B. T. Road, Kolkata 700108, India}	

	% REQUIRED
	\begin{abstract}
       The stability analysis of synchronization in time-varying higher-order networked structures (simplicial complexes) is one of the challenging problem due to the presence of time-varying group interactions. In this context, most of the previous studies have been done either on temporal pairwise networks or on static simplicial complexes. Here, for the first time, we propose a general framework to study the synchronization phenomenon in temporal simplicial complexes. We show that the synchronous state exists as an invariant solution and obtain the necessary condition for it to be emerged as a stable state in fast switching regime. We prove that the time-averaged simplicial complex plays the role of synchronization indicator whenever the switching among simplicial topologies are adequately fast. We attempt to transform the stability problem into a master stability function form. Unfortunately, for the general circumstances, the dimension reduction of the master stability equation is cumbersome due to the presence of group interactions. However, we overcome this difficulty in two interesting situations based on either the functional forms of the coupling schemes or the connectivity structure of the simplicial complex, and demonstrate that the necessary condition mimics the form of a master stability function in these cases. We verify our analytical findings by applying them on synthetic and real-world networked systems. In addition, our results also reveal that with sufficient higher-order coupling and adequately fast rewiring, the temporal simplicial complex achieves synchrony even in a very low connectivity regime.           		
	\end{abstract}
      \maketitle	
	\section{Introduction}
	 In the last two decades, networks (graphs) have emerged as a fundamental tool to describe various dynamical process on many real-world and man-made complex systems \cite{newman2003structure,boccaletti2006complex}. However, this network based representation rely on a strong assumption that the unitary elements of a system interact with each other solely through pairwise links. When investigating specific sorts of phenomena (specifically linear processes), such assumption may indeed be appropriate, but it fails significantly short in accurately describing many other situations including nonlinear processes in social interactions \cite{benson2016higher}, brains \cite{functional_brain1,structural_brain}, protein interaction networks \cite{protein}, and collaboration networks \cite{coauthor1}. In all these complex systems, group interactions between unitary elements are omnipresent and simply can not be described utilizing the pairwise network description \cite{beyond_pairwise,majhi2022dynamics}.
	 \par In another development, it emerged that the static network formulation, where the connections between unitary components remain the same over time, has several limitations in describing many real-world scenarios. For example, in social networks \cite{wasserman1994social,skufca2004communication}, in various artificial and neuronal networks \cite{motter2013spontaneous,valencia2008dynamic,bassett2011dynamic}, the underlying connection topologies evolve in time due to continuous creation or termination of the pairwise links.    
	 \par As a result, the classical network representation has been expanded in many ways including simplicial complexes \cite{giusti2016two,aleksandrov1998combinatorial,hypergraph1} that takes into account higher-order interactions incorporating three or even more nodes concurrently, and temporal networks \cite{holme2012temporal} whose connection topologies alter as time passes. Naturally, the more accurate descriptive capability of these generic structures inevitably leads to increasing analytical complexity when analyzing various dynamical phenomena that emerge on them. One such intriguing collective phenomenon is synchronization \cite{syn_book,synchronization2,synchronization3}, where the system components gradually align themselves into a single temporal evolution.
	 \par Until now, investigations of synchronization have been done either in static simplicial complexes \cite{simplicialsync2,simplicialsync3,simplicialsync4,simplicialsync6,anwar2022intralayer,skardal2021higher}, or in temporal networks \cite{ghosh2022synchronized,intra2,rakshit2020intralayer,sinha_tv,faggian2019synchronization,sar2022swarmalators,porfiri2017memory} where group interactions are not taken into consideration. The inclusion of static many-body interactions have been associated with several novel phenomenon including abrupt synchronization transitions \cite{simplicialsync4,kachhvah2022hebbian,skardal2020higher}, multistability \cite{xu2020bifurcation}, chaos \cite{sun2022triadic}, and chimera states \cite{srilena_chimera}. However, most of these studies on synchronization in static simplicial complexes have so far been limited to different extended Kuramoto model \cite{beyond_pairwise}, and very few have been considered for general chaotic systems. In this regard, to investigate analytical stability of synchronization in static simplicial complexes, generalization of master stability function (MSF) approach \cite{msf} has been proposed very recently in \cite{simplicialsync2} and \cite{anwar2022stability}. On the other hand, to investigate the stable synchronization state in temporal pairwise networks, fast switching stability approach is proposed in \cite{stilwell2006sufficient,porfiri2009global}. Other than fast-switching criterion, the stability of synchronization state in pairwise temporal networks can also be achieved by connection graph stability method \cite{belykh2004blinking,belykh2004connection} and simultaneous block diagonalization scheme \cite{intra2,zhang2021unified}. Nevertheless, how the interplay between temporal networks and higher-order structures (simplicial complexes) influences the dynamical processes on interconnected complex systems is still in its infancy \cite{kachhvah2022first,schaub2,chowdhary2021simplicial}, specifically the analytical investigation of synchronization on temporal higher-order networks is still unexplored. This motivates us greatly to investigate the synchronization phenomena on temporal higher-order networks.  
	 \par We here propose the most generic framework to study dynamical processes in temporal simplicial complexes. Specifically, we take into account a group of entirely general (but identical) evolving systems that are arranged on the nodes of an arbitrary higher dimensional simplicial complex, and communicate with one another through many-body interactions and generic coupling schemes (in arbitrary linear or nonlinear form). Further, the pairwise links and other higher-order connections are allowed to rewire stochastically in time. In such broad circumstances, we demonstrate that the complete synchronization state exists as an invariant solution as long as either the coupling functions or the interaction topologies satisfy certain conditions. We thereafter prove that the time-averaged simplicial complex plays the role of synchronization indicator in sufficiently fast switching regime, which generalizes the result of Ref. \cite{stilwell2006sufficient} to the case of group interactions. Using the time-averaged dynamics, we derive the necessary condition for stable synchronization solution for adequately quick switching. The obtained master stability equation is a coupled linear differential equation due the presence of additional intricacy in terms of temporal connections and arbitrary higher-order group interactions. However, our approach mimics the MSF scheme and provides fully decoupled master stability equation with dimension equal to the dimension of unitary dynamical components for couple of instances: (i) when the functional forms of pairwise and higher-order couplings are generalized diffusive with an additional requirement irrespective of the choice of connectivity topology of the simplicial complex; and (ii) when the time-averaged Laplacian corresponding to pairwise interactions commutes with other generalized time-averaged Laplacians attributing to group interactions of different order irrespective of the form of coupling schemes. Thereafter, we verify the obtained analytical results by numerical simulations in real-world and synthetic temporal simplicial complexes whose individual node dynamics are given by three different paradigmatic chaotic systems. Besides, we also  demonstrate that for temporal simplicial complex with random connections between the individuals, synchronization state is achievable even in very low connectivity regime when the higher-order coupling is sufficient and rewiring among the structures of the simplicial complex happens fast. 
	 \par The remainder of this article is organized as follows. In Sec. \ref{model}, we introduce a general mathematical model for  time-varying simplicial complexes. Section \ref{analytical} illustrates all our theoretical results regarding invariance and stability of the synchronization state. Developing these analytical results, in Sec. \ref{numerical} we present the numerical outcomes on coupled dynamical systems. Lastly, in Sec. \ref{conclusion}, we conclude by discussing our results and future study.
	 \begin{figure*}[ht] 
	 	\centerline{
	 		\includegraphics[scale=0.3]{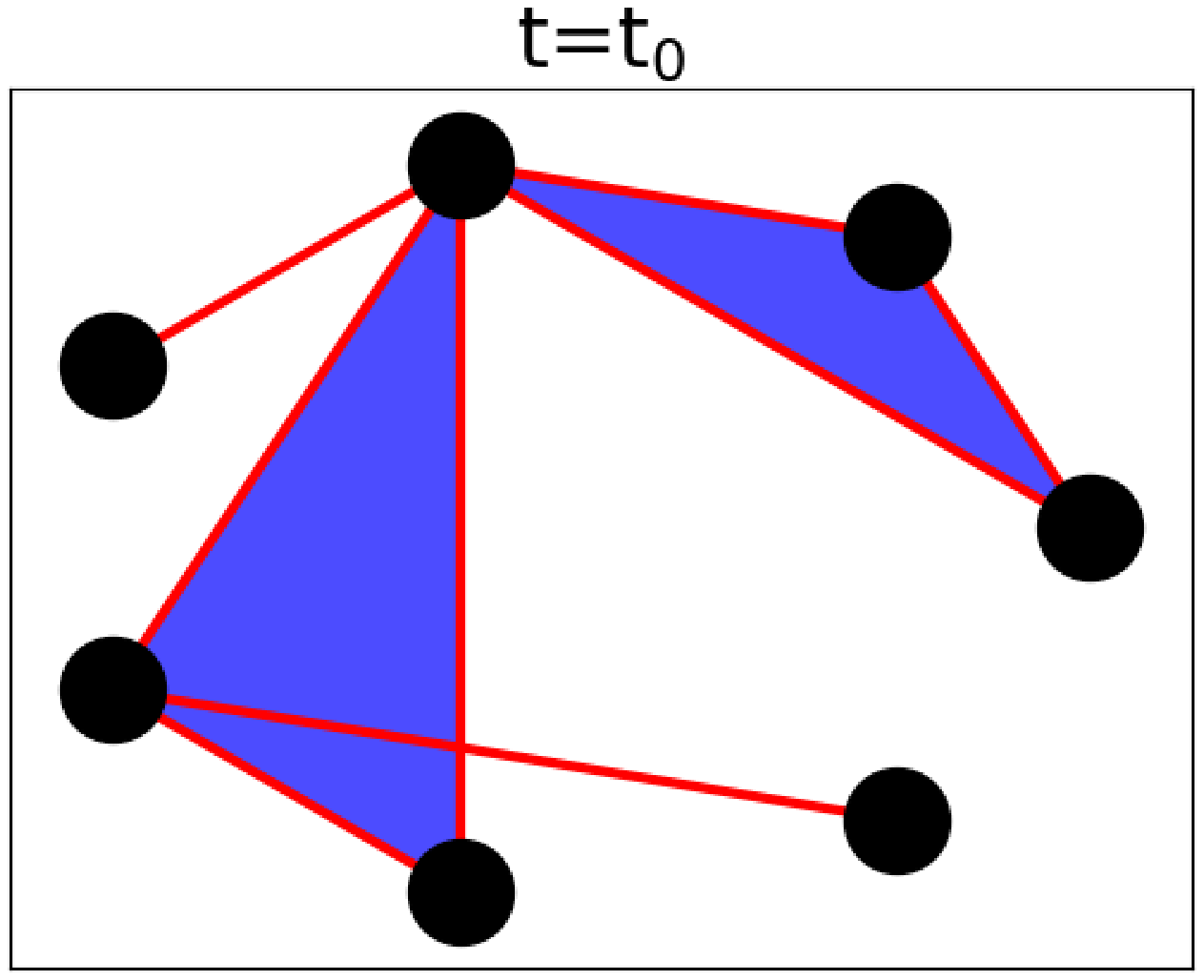}
	 		\includegraphics[scale=0.3]{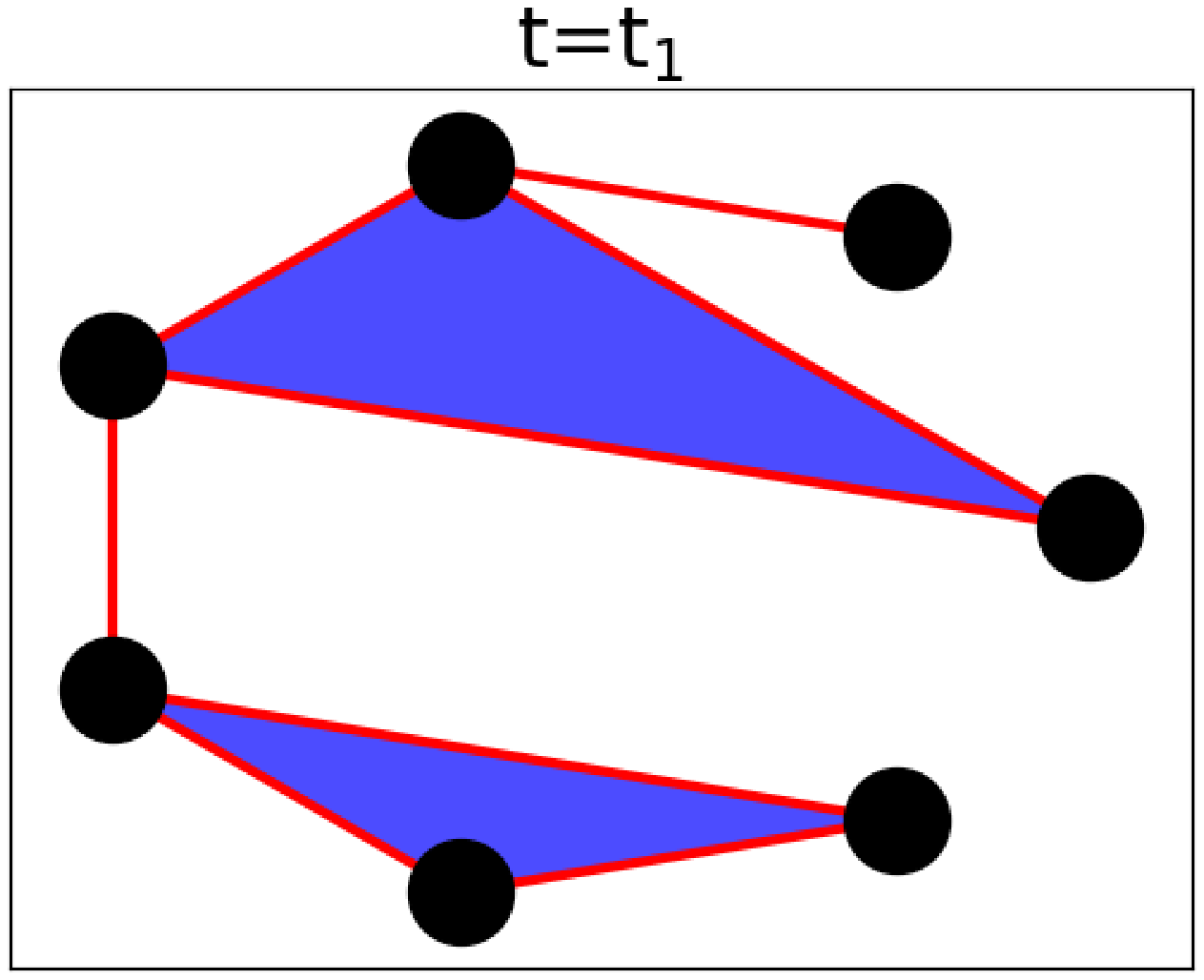}
	 		\includegraphics[scale=0.3]{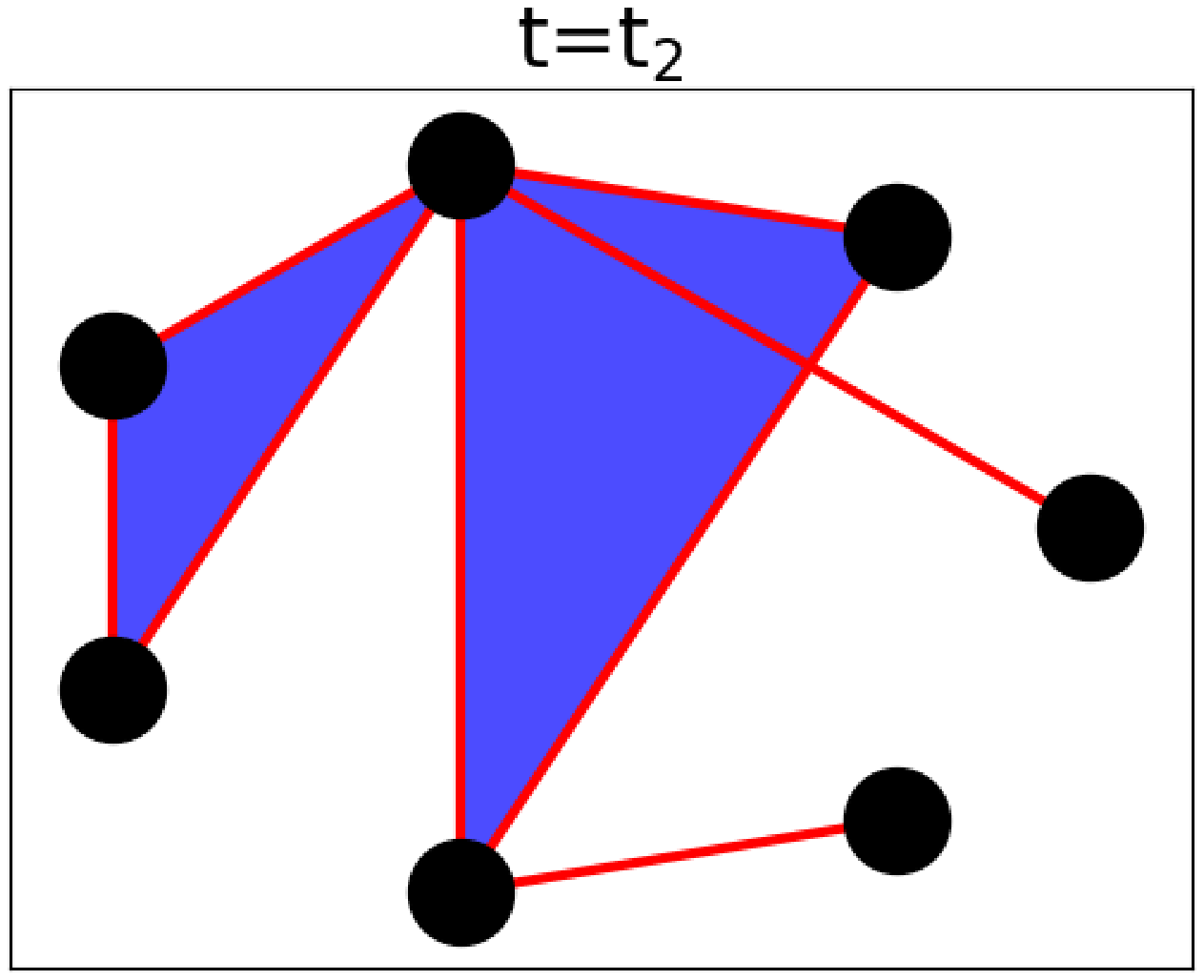}}
	 	\caption{ {\bf Schematic illustration of time-varying simplicial complex with pairwise and triadic interactions.}The left, middle and right panels represent the structure of simplicial complex at time stamps $t=t_{0},t_{1}, \mbox{and}\; t_{2}$, respectively. The solid red lines portrays the pairwise connections ($1$-simplices) and shaded blue triangles represent three-body interactions ($2$-simplices).}
	 	\label{schematic}
	 \end{figure*}                                   
     \section{General model for temporal simplicial complex} \label{model}
     As the objective of our work is to propose a general model that account for group interaction of any order between the dynamical units at any instance of time, we consider a $P$-dimensional simplicial complex composed of $N$ number of nodes. This implies that the nodes interact not only via pairwise links, but also through $2$-simplices (triangles), $3$-simplices (tetrahedrons) and so on. Then, The dynamics of the $P$-dimensional simplicial complex are governed by the following equation of motion 
     \begin{widetext}
     	\begin{equation}\label{gen_model}
     	\begin{array}{l}
     		\dot{\bf x}_i(t) = f({\bf x}_i)+\epsilon_{1} \sum_{k_1=1}^{N} \mathscr{B}^{(1)}_{ik_1}(t) H^{(1)}({\bf x}_i,{\bf x}_{k_1})  
     		+ \epsilon_{2} \sum_{k_1=1}^{N} \sum_{k_2=1}^{N} \mathscr{B}^{(2)}_{ik_1k_2}(t)H^{(2)}({\bf x}_i,{\bf x}_{k_1},{\bf x}_{k_2})  + \cdots+ \\\\  ~~~~~~~~~~~~~~  \epsilon_{P} \sum_{k_1=1}^{N} \sum_{k_2=1}^{N}\cdots \sum_{k_P=1}^{N} \mathscr{B}^{(P)}_{ik_1k_2\cdots k_P}(t)H^{(P)}({\bf x}_i,{\bf x}_{k_1},{\bf x}_{k_2},\cdots,{\bf x}_{k_P}),
     	\end{array}
     \end{equation}
     \end{widetext}
     where the $D$-dimensional state vector ${\bf x}_i(t)$ denotes the dynamics of node $i$, $f : \mathbb{R}^D \rightarrow  \mathbb{R}^D$ provides an illustration of the individual node dynamics assumed to be the same for all the dynamical units. $\epsilon_{m}$ ($m=1,2,\cdots,P$) are the real valued constants describing coupling strengths associated with $m$-dimensional simplices. Here $H^{(m)}:  \mathbb{R}^{((m+1) \times D)} \rightarrow  \mathbb{R}^D$ $(m=1,2,\cdots,P)$ are continuously differentiable real valued functions accounting for coupling schemes of various orders. Further, the elements $\mathscr{B}^{(m)}_{ik_1k_2\cdots k_m}(t)$ of the adjacency tensors $\mathscr{B}^{(m)}(t)$ $(m=1,2,\cdots, P)$ describe the connection mechanism between the dynamical units at a time instance $t$. $\mathscr{B}^{(m)}_{ik_1k_2\cdots k_m}(t)=1$, if the set of nodes $\{i,k_1,k_2,\cdots,k_m\}$ are connected through a $m$-dimensional simplex at a time stamp $t$ and zero otherwise. Here, the adjacency tensors are functions of time, i.e., the pairwise and higher-order connections between the nodes vary in time through the stochastic alteration of the entire connection topology with a rewiring frequency $r$. Sufficiently smaller rewiring frequency indicates that the connections between the nodes are almost static in time, whereas comparably higher $r$ indicates very fast swapping of connections between the nodes. To be noted, the proposed system matches the typical scenario of a temporal complex network of $N$ connected dynamical units in the specific case of $P = 1$. As there are no special limitations on the functions $f$, $H^{m}$, and the time-dependent adjacency tensors $\mathscr{B}^{(m)}$ of the simplicial complex, this is the most generic form of system we can take into consideration. For example, a diagrammatic representation of $2$-dimensional temporal simplicial complex composed of $N=7$ nodes is represented in Fig. \ref{schematic}. The nodes interact with each other through pairwise links (solid red lines) and $2$-simplices (shaded blue triangles) at different time instances. The left, middle and right panel display the evolution of pairwise and three-body interaction over three different time stamps, respectively.
     We further define various generalized $m$-degree of a node $i$ at any time instance $t$ corresponds to $m$-simplices by denoting $\frac{1}{m!}{\sum\limits_{k_1=1}^{N}\sum\limits_{k_2=1}^{N}\cdots\sum\limits_{k_m=1}^{N}{\mathscr{B}_{ik_1k_2\cdots k_m}^{(2)}}}(t)$ as $d_{i}^{(m)}(t)$. For instance, $\sum_{j=1}^{N} \mathscr{B}^{(1)}_{ij}(t)=d_{i}^{(1)}(t)$ denotes the number of pairwise links ($1$-simplices) attached to node $t$ at the time point $t$, and $\frac{1}{2}{\sum_{j=1}^{N}\sum_{k=1}^{N} \mathscr{B}^{(2)}_{ijk}(t)=d_{i}^{(2)}(t)}$ represents the number of $2$-simplices to which node $i$ takes part at time instance $t$.      
     \section{Analytical results}\label{analytical}
     Throughout this section, the main object of our investigation is to acquire the necessary condition for stable synchronization solution in the time-varying simplicial complex \eqref{gen_model}. To accomplish this, we initially investigate whether the existence of a complete synchronized state can be guaranteed for the proposed simplicial complex \eqref{gen_model}.  
     \subsection{Existence condition}\label{invariance}
      The simplicial complex \eqref{gen_model} achieves a complete synchronization state when each node advances synchronously with the other nodes. Mathematically, there exists a solution $\mathbf{x}_0(t) \in \mathbb{R}^{D}$ such that, 
     \begin{equation} \label{comp_sync_def}
     	\begin{array}{l}
     		\norm{\mathbf{x}_j(t)-\mathbf{x}_0(t)} \to 0 ~~ \mbox{as}~~ t \to \infty ~~\mbox{for}~~ j=1,2,\dots,N.
     	\end{array}		
     \end{equation} 
     Following this, the corresponding synchronized manifold can be defined as follows, 
     \begin{equation}\label{manifold}
     	\begin{array}{l}
     		\mathcal{M}=\{\mathbf{x}_0(t) \subset \mathbb{R}^{D} : \mathbf{x}_j(t)=\mathbf{x}_0(t), ~~ \mbox{for} ~~ j=1,2,\dots,N 
     		\\ ~~~~~~~~~~~~~~~~~~~~ \mbox{and} ~~ t \in \mathbb{R}^+ \}.	
     	\end{array} 
     \end{equation} 
     In general, a complete synchronous state is not always achievable for the proposed simplicial complex due to the presence of arbitrary coupling scheme between the unitary dynamical units. Therefore, we consider a specific form of the coupling functions which vanishes at the synchronous state, called synchronization noninvasive coupling functions. If the functional forms of coupling schemes are synchronization noninvasive, i.e.,  
     \begin{equation} \label{noninvasive}
     	\begin{array}{l}
     		H^{(1)}({\bf x}_0,{\bf x}_0)=H^{(2)}({\bf x}_0,{\bf x}_0,{\bf x}_0)= \cdots 
     	  \\	~~~~~~~~~~~~~~~ = H^{(P)}({\bf x}_0,{\bf x}_0,{\bf x}_0,\cdots,{\bf x}_0)=0.	
     	\end{array}
     \end{equation} 
     then the existence and invariance of the synchronization manifold $\mathcal{M}$ will be guaranteed since each individual node follows same dynamics given by, 
     \begin{equation}\label{sync_sol1}
     	\begin{array}{l}
     		\dot{\bf x}_0(t) = F({\bf x}_0).
     	\end{array}
     \end{equation}
     At this point, it is important to note that by considering the synchronization noninvasive coupling functions, we are addressing a large class of coupling schemes. For example, generalized diffusive coupling functions, diffusive sine couplings used for coupled Kuramoto oscillators are some specific cases of this noninvasive coupling form.    
     \par Further, the invariance of synchronization state can also be guaranteed for any arbitrary form of coupling functions apart from noninvasive form. The discussion of invariance condition for generic coupling schemes is detailed in Appendix \ref{generic coupling}. 
     
     \subsection{Linear stability analysis}\label{stability analysis}
     Now since it does not take a very strong force to disrupt the synchronization state, we look into the local stability of the synchronization solution. To accomplish this, we consider a small perturbation around the synchronization solution ${\bf x}_0$, i.e., $\delta {\bf x}_i= {\bf x}_i-{\bf x}_0$ and perform the linear stability analysis. As the expressions could appear a bit complicated, we will use the case of $P = 2$-dimensional simplicial complex in the following to avoid complexity in the notations. Thereafter, we will extrapolate the findings to all values of $P$. Hence, for $2$-dimensional simplicial complex, the linearized equation in terms of the stake variables $\delta {\bf x}_i$ can be written as,
     \begin{widetext} 
     \begin{equation} \label{stability_1}
     	\begin{array}{l}
     		\delta \dot{\bf x}_{i}= Jf({\bf x}_{0}) \delta {\bf x}_{i} + \epsilon_{1} \sum_{j=1}^{N} \mathscr{B}^{(1)}_{ij}(t) \big[H^{(1)}_{{\bf x}_{i}}({\bf x}_{0},{\bf x}_{0}) \delta {\bf x}_{i}+ H^{(1)}_{{\bf x}_{j}}({\bf x}_{0},{\bf x}_{0}) \delta {\bf x}_{j} \big] \\\\
     		~~~~~~~~~~ + \epsilon_{2} \sum_{j=1}^{N} \sum_{k=1}^{N}\mathscr{B}^{(2)}_{ijk}(t) \big[H^{(2)}_{{\bf x}_{i}}({\bf x}_{0},{\bf x}_{0},{\bf x}_{0}) \delta {\bf x}_{i}+ H^{(2)}_{{\bf x}_{j}}({\bf x}_{0},{\bf x}_{0},{\bf x}_{0}) \delta {\bf x}_{j} 
     		+ H^{(2)}_{{\bf x}_{k}}({\bf x}_{0},{\bf x}_{0},{\bf x}_{0}) \delta {\bf x}_{k} \big],
     	\end{array}
     \end{equation}
     \end{widetext}
     where $Jf({\bf x}_{0})$ is the Jacobian matrix of $f$ evaluated at the synchronization solution ${\bf x}_{0}$. $H^{(m)}_{{\bf x}_{i}}$ $(m=1,2)$ is the derivative of $H^{(m)}$ with respect to the variable ${{\bf x}_{i}}$. Now to do further analysis, we proceed with the noninvasive coupling functions. The case for arbitrary linear or nonlinear form of coupling functions is discussed in the Appendix \ref{generic coupling}. As for the synchronization noninvasive form, the coupling functions vanish at the synchronous state, it immediately implies that the value of their total derivative also becomes null at the synchronization solution. Mathematically, 
     \begin{equation}\label{total_derivative}
     	\begin{array}{l}
     		H^{(1)}_{{\bf x}_{i}}({\bf x}_{0},{\bf x}_{0})+ H^{(1)}_{{\bf x}_{j}}({\bf x}_{0},{\bf x}_{0})=0 
     		,\mbox{and} \\H^{(2)}_{{\bf x}_{i}}({\bf x}_{0},{\bf x}_{0},{\bf x}_{0})+ H^{(2)}_{{\bf x}_{j}}({\bf x}_{0},{\bf x}_{0},{\bf x}_{0}) + H^{(2)}_{{\bf x}_{k}}({\bf x}_{0},{\bf x}_{0},{\bf x}_{0})=0
     	\end{array}
     \end{equation}. 
     Using the relation \eqref{total_derivative} and the fact that $\sum_{j=1}^{N} \mathscr{B}^{(1)}_{ij}(t)=d_{i}^{(1)}(t)$ and $\sum_{j=1}^{N}\sum_{k=1}^{N} \mathscr{B}^{(2)}_{ijk}(t)=2d_{i}^{(2)}(t)$, the linearized Eq. \eqref{stability_1} can be rewritten as follows,
     \begin{widetext}
     \begin{equation} \label{stability_2}
     	\begin{array}{l}
     		\delta \dot{\bf x}_{i}= Jf({\bf x}_{0}) \delta {\bf x}_{i} - \epsilon_{1} \sum_{j=1}^{N} \mathscr{L}^{(1)}_{ij}(t) H^{(1)}_{{\bf x}_{j}}({\bf x}_{0},{\bf x}_{0}) \delta {\bf x}_{j} \\\\
     		~~~~~~~~~~ - \epsilon_{2} \sum_{j=1}^{N} \sum_{k=1}^{N}\mathscr{\tau}^{(2)}_{ijk}(t) \big[ H^{(2)}_{{\bf x}_{j}}({\bf x}_{0},{\bf x}_{0},{\bf x}_{0}) \delta {\bf x}_{j} + H^{(2)}_{{\bf x}_{k}}({\bf x}_{0},{\bf x}_{0},{\bf x}_{0}) \delta {\bf x}_{k} \big],
     	\end{array}
     \end{equation}
     \end{widetext}
     where $\mathscr{L}^{(1)}(t)$ is the classical graph Laplacian defined as $\mathscr{L}^{(1)}_{ij}(t)=d^{(1)}_{i}(t)\delta_{ij}-\mathscr{B}^{(1)}_{ij}(t)$, and $\tau^{(2)}_{ijk}(t)$ are the elements of the tensor $\mathscr{T}^{(2)}(t)$, given by the relation $\tau^{(2)}_{ijk}(t)=2d^{(2)}_{i}(t)\delta_{ijk}-\mathscr{B}^{(2)}_{ijk}(t)$. Further using the symmetric property of the elements $\tau^{(2)}_{ijk}(t)$ and the concept of generalized Laplacian $\mathscr{L}^{(2)}(t)$, defined as $\mathscr{L}^{(2)}_{ij}(t)=2d^{(2)}_{i}(t)\delta_{ij}-\sum_{k=1}^{N}\tau^{(2)}_{ijk}(t)$, one can simplify the linearized equation \eqref{stability_2} as 
     \begin{widetext}
     \begin{equation} \label{stability_6}
     	\begin{array}{l}
     		\delta \dot{\bf x}_{i}= Jf({\bf x}_{0}) \delta {\bf x}_{i} - \epsilon_{1} \sum_{j=1}^{N} \mathscr{L}^{(1)}_{ij}(t) H^{(1)}_{{\bf x}_{j}}({\bf x}_{0},{\bf x}_{0}) \delta {\bf x}_{j} \\\\
     		~~~~~~~~~~ - \epsilon_{2} \sum_{j=1}^{N} \mathscr{L}^{(2)}_{ij}(t) \big[ H^{(2)}_{{\bf x}_{j}}({\bf x}_{0},{\bf x}_{0},{\bf x}_{0}) + H^{(2)}_{{\bf x}_{k}}({\bf x}_{0},{\bf x}_{0},{\bf x}_{0})\big] \delta {\bf x}_{j}.
     	\end{array}
     \end{equation}
     \end{widetext}
     Thereafter we rewrite the Eq. \eqref{stability_6} in block matrix form by introducing stack vector $\delta \mathbf{X}= [\delta \mathbf{x}_{1}^{tr}, \delta \mathbf{x}_{2}^{tr}, \cdots, \delta \mathbf{x}_{N}^{tr}]^{tr}$, where $[\;\;]^{tr}$ denotes the vector transpose. Eventually the linearized equation in block matrix form becomes
     \begin{widetext}
     \begin{equation}\label{stability_7}
     	\begin{array}{l}
     		\delta \dot{\bf X}=[I_{N} \otimes Jf({\bf x}_{0})-\epsilon_{1}\mathscr{L}^{(1)}(t) \otimes H^{(1)}_{{\bf x}_{j}}({\bf x}_{0},{\bf x}_{0})
     		\\\\ ~~~~~~~~~~~~~~~~~~~~~ -\epsilon_{2} \mathscr{L}^{(2)}(t) \otimes \big\{H^{(2)}_{{\bf x}_{j}}({\bf x}_{0},{\bf x}_{0},{\bf x}_{0}) + H^{(2)}_{{\bf x}_{k}}({\bf x}_{0},{\bf x}_{0},{\bf x}_{0})\big\}] \delta {\bf X},
     	\end{array}
     \end{equation}
     \end{widetext}
     where $I_{N}$ represents the identity matrix of dimension $N$, and $\otimes$ symbolizes the tensor product.
     Further denoting $H^{(1)}_{{\bf x}_{j}}({\bf x}_{0},{\bf x}_{0})$ as $\mathcal{J}H^{(1)}$, $\{H^{(2)}_{{\bf x}_{j}}({\bf x}_{0},{\bf x}_{0},{\bf x}_{0}) + H^{(2)}_{{\bf x}_{k}}({\bf x}_{0},{\bf x}_{0},{\bf x}_{0})\big\}$ as $\mathcal{J}H^{(2)}$ and considering all the higher-order interactions for $P$-dimensional simplicial complex, we can extrapolate the linearized equation for $P$-dimensional simplicial complex in block form as,
     \begin{equation}\label{time_dsimp_linear}
     	\begin{array}{l}
     		\delta \dot{\bf X}=[I_{N} \otimes Jf({\bf x}_{0})-\sum_{m=1}^{P}\epsilon_{m}\mathscr{L}^{(m)}(t) \otimes \mathcal{J}H^{(m)}] \delta {\bf X},
     	\end{array}
     \end{equation}
     where $\mathscr{L}^{(m)}(t)$, $(m=1,2,\cdots,P)$ is the $N\times N$ generalized Laplacian matrix associated with $m$-simplex interaction and $\mathcal{J}H^{(m)}=H^{(m)}_{{\bf x}_{k_1}}({\bf x}_{0},{\bf x}_{0},\cdots,{\bf x}_{0}) + H^{(m)}_{{\bf x}_{k_2}}({\bf x}_{0},{\bf x}_{0},\cdots,{\bf x}_{0})+\cdots+H^{(m)}_{{\bf x}_{k_m}}({\bf x}_{0},{\bf x}_{0},\cdots,{\bf x}_{0})$ is the sum of partial derivatives corresponding to the coupling function $H^{(m)}({\bf x}_i,{\bf x}_{k_1},\cdots,{\bf x}_{k_m})$, $(m=1,2,\cdots,P)$.
     \par The variational equation \eqref{time_dsimp_linear} include two parts: parallel modes, which describe motion along the synchronization manifold, and transverse modes, which describe motion across the synchronization manifold. To achieve a stable synchronization solution, all transverse modes must finally extinct in time. The variational equation should then be decoupled into parallel and transverse modes in order to conduct a linear stability analysis and determine whether the latter ones die out or not. The conventional MSF scheme properly decouples the high-dimensional variational equation to perturbation mode independent low-dimensional equations through an appropriate coordinate transformation that completely diagonalizes the coupling matrices in order to achieve a stable synchronization state for static pairwise network structure. Even for steady simplicial complexes due to the presence of higher-order non-commuting generalized Laplacians, it is not always possible to decouple the variational equation into low-dimensional equations through an appropriate coordinate transformation \cite{simplicialsync2}. As a result, the problem of stability becomes too expensive. Therefore, when one considers the time-varying simplicial complexes, the presence of non-commuting generalized Laplacian matrices over each time stamps create even more difficulty to tackle the stability problem. As a consequence, in the temporal scenario, the traditional MSF formalism can not be directly applicable to simultaneously diagonalize the generalized Laplacian matrices through a suitable coordinate transformation. To overcome this difficulty, we use the fast-switching stability criterion.
     \par Fast switching between the time-frozen network configurations implies that the evolution of network structure is occurring at a rate that is much faster than that of the underlying dynamical process. For the pairwise network structure, it is well established that in the fast switching limit the temporal network structure shows asymptotically stable synchronization manifold whenever the corresponding static time-averaged network has an asymptotically stable synchronized solution \cite{stilwell2006sufficient}. We here extend this result to higher-order network structures of any order. Specifically we prove that for time-varying simplicial complexes of any order $P$ the corresponding time-averaged structure plays the role of synchronization indicator in fast-switching limit. Our result states as follows,
     \begin{theorem} \label{th1}
     	The time varying simplicial complex whose dynamic evolution given by \eqref{gen_model} has an asymptotically stable synchronized solution for sufficiently fast switching whenever the corresponding static time-averaged simplicial complex whose dynamics is given by
     	\begin{widetext}
     	\begin{equation}\label{avg_gen_model}
     		\begin{array}{l}
     			\dot{\bar{{\bf x}}}_i(t) = f(\bar{{\bf x}}_i)+\epsilon_{1} \sum_{k_1=1}^{N} \bar{\mathscr{B}}^{(1)}_{ik_1}(t) H^{(1)}(\bar{{\bf x}}_i,\bar{{\bf x}}_{k_1})  +  \epsilon_{2} \sum_{k_1=1}^{N} \sum_{k_2=1}^{N} \bar{\mathscr{B}}^{(2)}_{ik_1k_2}(t)H^{(2)}(\bar{{\bf x}}_i,\bar{{\bf x}}_{k_1},\bar{{\bf x}}_{k_2}) \\\\ ~~~~~~~~~~~~~~ + \cdots+ \epsilon_{P} \sum_{k_1=1}^{N} \sum_{k_2=1}^{N}\cdots \sum_{k_P=1}^{N} \bar{\mathscr{B}}^{(P)}_{ik_1k_2\cdots k_P}(t)H^{(P)}(\bar{{\bf x}}_i,\bar{{\bf x}}_{k_1},\bar{{\bf x}}_{k_2},\cdots,\bar{{\bf x}}_{k_P})
     		\end{array}
     	\end{equation} 
     	is asymptotically stably synchronized.
     	\end{widetext}	
     \end{theorem}   
     Here $\bar{{\bf x}}_{i}$, $(i=1,2,\cdots,N)$ represent the dynamics of the nodes in static time-averaged simplicial complex. $\bar{\mathscr{B}}^{(m)}$, $(m=1,2,\cdots,P)$ are the static time-averaged adjacency tensors corresponding to the time-varying adjacency tensors ${\mathscr{B}}^{(m)}(t)$ can be defined for a real constant $T$ by 
     \begin{equation}\label{avg_adj}
     	\begin{array}{l}
     		\bar{\mathscr{B}}^{(m)}=\frac{1}{T} \bigintsss_{t}^{t+T} \mathscr{B}^{(m)}(z)dz.
     	\end{array}
     \end{equation} 
     Subsequently each time-varying generalized Laplacian matrix possesses a time-averaged static Laplacian given by 
     \begin{equation}\label{avg_lap}
     	\begin{array}{l}
     		\bar{\mathscr{L}}^{(m)}=\frac{1}{T} \bigintsss_{t}^{t+T} \mathscr{L}^{(m)}(z)dz, ~~m=1,2,\cdots,P.
     	\end{array}
     \end{equation}  
     \begin{remark}
     	The time-averaged Laplacians $\bar{\mathscr{L}}^{(m)}$, $(m=1,2,\cdots,P)$ do not necessarily correspond to the generalized Laplacian matrices of a simplicial complex. These matrices only arise as a time-average of each time-varying generalized Laplacians $\mathscr{L}^{(m)}(t)$. Nevertheless, these time-averaged Laplacians inherit the zero-row sum property of the generalized Laplacians $\mathscr{L}^{(m)}(t)$. 
     \end{remark}  
     Now to proof the theorem \ref{th1}, we require the following lemma proposed by Stilwell et al. \cite{stilwell2006sufficient} for fast-switching stability protocol.
     \begin{lemma} \label{lemma1}	
     	Assume a constant T for which the matrix valued function $\mathscr{E}(t)$ satisfies
     	\begin{equation*}
     		\begin{array}{l}
     			\bar{\mathscr{E}}=\frac{1}{T} \bigintsss_{t}^{t+T}{\mathscr{E}(\tau)} d\tau,~\forall~t \in \mathbb{R}^{+},
     		\end{array}
     	\end{equation*} 
     	and 
     	\begin{equation}\label{eq1.lemma1}
     		\begin{array}{l}
     			\dot{\bar{\mathbf{y}}}(t)=[B(t)+\bar{\mathscr{E}}] \bar{\mathbf{y}}(t), \hspace{10pt}\bar{\mathbf{y}}(t_s)=\bar{\mathbf{y}}_s, \hspace{10pt}t \ge t_s,
     		\end{array}
     	\end{equation}
     	is uniformly exponentially stable. Then for sufficiently fast switching 
     	\begin{equation}{\label{eq2.lemma1}}
     		\begin{array}{l}
     			\dot{\mathbf{y}}(t)=[B(t)+\mathscr{E}(t)] \mathbf{y}(t),\hspace{10pt}\mathbf{y}(t_s)=\mathbf{y}_s,\hspace{10pt}t\ge t_s,
     		\end{array}
     	\end{equation}
     	will also be uniformly asymptotically stable. 	
     \end{lemma}

     \begin{proof}[Proof of theorem \ref{th1}]
     	\par Comparing Eqs. \eqref{gen_model} and \eqref{avg_gen_model}, one can easily obtain that for both the time-averaged and time-varying simplicial complexes, the synchronization solution follows same evolution equation given by \eqref{sync_sol1}. Hence without losing the generality we can consider that, at the state of synchrony $\bar{{\bf x}}_i={\bf x}_{0}$, for all $i=1,2,\cdots,N$. Now to study the stability of synchronization for the time-averaged simplicial complex, we assume a small perturbation $\delta \bar{\mathbf{X}}= [\delta \bar{\mathbf{x}}_{1}^{tr}, \delta \bar{\mathbf{x}}_{2}^{tr}, \cdots, \delta \bar{\mathbf{x}}_{N}^{tr}]^{tr}$ around the synchronization solution, defined as $\delta \bar{\mathbf{x}}_{i}=\bar{\mathbf{x}}_{i}-{\mathbf{x}}_{0}$, $(i=1,2,\cdots,N)$. Then for the time-averaged simplicial complex, the variational equation becomes,
     	\begin{equation}\label{avg_dsimp_linear}
     		\begin{array}{l}
     			\delta \dot{\bar{{\bf X}}}=[I_{N} \otimes Jf({\bf x}_{0})-\sum_{m=1}^{P}\epsilon_{m}\bar{\mathscr{L}}^{(m)} \otimes \mathcal{J}H^{(m)}] \delta \bar{{\bf X}}.
     		\end{array}
     	\end{equation}       
     	Now all the time-averaged Laplacians $\bar{\mathscr{L}}^{(m)},\; (m=1,2,\cdots,P)$ are real symmetric matrices satisfying the zero-row sum property. As a result, they all have the same least eigenvalue $\lambda_{1}^{(m)}=0$, with corresponding eigenvector $v_{1}^{(m)}=\frac{1}{\sqrt{N}}(1,1,\cdots,1)^{tr}$ aligned along the synchronization manifold. Again all the time averaged Laplacians are orthonormally diagonalizable by the basis of their eigenvectors $V^{(m)}=[v_{1}^{(m)},v_{2}^{(m)},\cdots,v_{N}^{(m)}]$, i.e., there exists a diagonal matrix $\bar{U}^{(m)}$ such that $V^{(m)^{-1}}\bar{\mathscr{L}}^{(m)}V^{(m)}= \bar{U}^{(m)}$, where the diagonal elements of $\bar{U}^{(m)}$ are the eigenvalues of the time-averaged Laplacian matrices $\bar{\Lambda}^{(m)}=\{0=\lambda_{1}^{(m)}< \lambda_{2}^{(m)}\leq \lambda_{3}^{(m)} \le \cdots \le \lambda_{N}^{(m)}\}$.  
     	\par The variational equation \eqref{avg_dsimp_linear} associated with the time-averaged structure contains both the transverse and parallel modes to the synchronization manifold. To analyze the stability of synchronization manifold, one needs to separate the transverse modes from the parallel ones and scrutinize whether the transverse modes extinct in time or not. We, therefore, project the stake variable $\delta\bar{{\bf X}}$ to the basis of eigenvectors $V^{(1)}$ associated with the time-averaged Laplacian $\bar{\mathscr{L}}^{(1)}$ by introducing new variable $\bar{\zeta}=(V^{(1)}\otimes I_{D})^{-1}\delta \bar{\bf X}$. However, the choice for the basis of eigenvectors is completely arbitrary, one can choose any of the time-averaged generalized Laplacian as a reference for this choice, and all other set of eigenvectors can be mapped to such a basis by means of unitary matrix transformation. The variational equation \eqref{avg_dsimp_linear} in terms of new variable $\zeta$ can then be written as follows,
     	\begin{equation}\label{zeta_dsimp_linear}
     		\begin{array}{l}
     			\dot{\bar{{\zeta}}}=[I_{N} \otimes Jf({\bf x}_{0})-\sum_{m=1}^{P}\epsilon_{m}(V^{(1)^{-1}}\bar{\mathscr{L}}^{(m)}V^{(1)}) \otimes \\ ~~~~~~~~~~~~~~~~~~~~~~~~~~~~~~~~~~~~~~~~~~~~~~~~~ \mathcal{J}H^{(m)}] \bar{\zeta}.
     		\end{array}
     	\end{equation}
     	Now using the orthogonality condition of the eigenvector sets and the fact that \\ $V^{(m)^{-1}}\bar{\mathscr{L}}^{(m)}V^{(m)}=\bar{U}^{(m)}$, we can deduce that $V^{(1)^{-1}}\bar{\mathscr{L}}^{(m)}V^{(1)}$, for all $m=1,2,\cdots,P$, is a real square matrix with null first column and first row, i.e.,   
     	\begin{equation}\label{avg_triangular_2}
     		\begin{array}{l}
     			V^{(1)^{-1}}\bar{\mathscr{L}}^{(m)}V^{(1)}= \bar{W}^{(m)}= 
     			\begin{bmatrix}
     				0 & 0_{1\times (N-1)} \\
     				0_{(N-1)\times 1} & \bar{W}^{(m)}_{2}
     			\end{bmatrix}.
     		\end{array}
     	\end{equation}
     	Plugging back \eqref{avg_triangular_2} in \eqref{zeta_dsimp_linear}, one eventually obtains that the transformed variable $\bar{\zeta}$ decomposes into two independent parts $\bar{\zeta}=[\bar{\zeta}_{\parallel} \;\; \bar{\zeta}_{\perp}]$, where $\bar{\zeta}_{\parallel} \in \mathbb{R}^{D}$ and $\bar{\zeta}_{\perp} \in \mathbb{R}^{(N-1)\times D}$, and their dynamics can be represented as follows, 
     	\begin{subequations}
     		\begin{equation}\label{avg_parallel}
     			\dot{\bar{\zeta}}_{\parallel}=Jf({\bf x}_{0}){\bar{\zeta}_{\parallel}},
     		\end{equation}
     		\begin{equation}\label{avg_transverse}
     			\dot{\bar{\zeta}}_{\perp}= [I_{N-1} \otimes Jf({\bf x}_{0})-\sum_{m=1}^{P}\epsilon_{m}\bar{W}^{(m)}_{2} \otimes \mathcal{J}H^{(m)}] \bar{\zeta}_{\perp}.	
     		\end{equation}
     	\end{subequations}
     	Here $\bar{{\zeta}}_{\parallel}$ reflects the dynamics of the perturbation modes along the synchronization manifold, called parallel modes, and the dynamics of $\bar{{\zeta}}_{\parallel}$ represents the evolution of the perturbation modes across the synchronous manifold, called transverse modes. By the hypothesis of theorem \ref{th1}, the time-averaged simplicial complex \eqref{avg_gen_model} possesses a stable synchronized manifold. Therefore, the transverse variational equation \eqref{avg_transverse} is asymptotically stable. 
     	\par Now we apply the same coordinate transformation to the variational equation \eqref{time_dsimp_linear} associated with temporal simplicial complex as ${\zeta}=(V^{(1)}\otimes I_{D})^{-1}\delta{\bf X}$. Similar to the time-averaged system, in this scenario also the transformed variable decomposes into parallel and transverse components ${\zeta}=[{\zeta}_{\parallel} \;\; {\zeta}_{\perp}]$, whose dynamics are given by the following set of equations
     	\begin{subequations}
     		\begin{equation}\label{time_parallel}
     			\begin{array}{l}
     				\dot{{\zeta}}_{\parallel}= Jf({\bf x}_{0}) {\zeta}_{\parallel},
     			\end{array}
     		\end{equation}		
     		\begin{equation}\label{time_trans}
     			\begin{array}{l}
     				\dot{{\zeta}}_{\perp}= [I_{N-1} \otimes Jf({\bf x}_{0})-\sum_{m=1}^{P}\epsilon_{m}{W}^{(m)}_{2}(t) \otimes \mathcal{J}H^{(m)}] {\zeta}_{\perp},
     			\end{array}
     		\end{equation}   
     	\end{subequations}
     	where ${W}^{(m)}_{2}(t)$ $(m=1,2,\cdots,P)$ are $(N-1)\times(N-1)$ time-varying matrices satisfying the relation
     	\begin{equation}\label{time_triangular_2}
     		\begin{array}{l}
     			V^{(1)^{-1}}{\mathscr{L}}^{(m)}(t)V^{(1)}= {W}^{(m)}(t)= 
     			\begin{bmatrix}
     				0 & 0_{1\times(N-1)} \\
     				0_{(N-1)\times 1} & {W}^{(m)}_{2}(t)
     			\end{bmatrix}.
     		\end{array}
     	\end{equation} 
     	Now, from the relation between time-averaged and time-varying Laplacians \\ $\bar{\mathscr{L}}^{(m)}=\frac{1}{T} \bigintsss_{t}^{t+T} \mathscr{L}^{(m)}(z)dz$, and using \eqref{avg_triangular_2} and \eqref{time_triangular_2}, we can derive that 
     	\begin{equation}\label{avg_relation}
     		\begin{array}{l}
     			\bar{W}^{(m)}_{2}=\frac{1}{T} \bigintsss_{t}^{t+T} {W}^{(m)}_{2}(z)\;dz.	
     		\end{array}
     	\end{equation}      
     	Therefore, applying Lemma \ref{lemma1} along with \eqref{avg_transverse}, \eqref{time_trans} and \eqref{avg_relation}, we can immediately conclude  that the time-varying simplicial complex will also possesses stable synchronization manifold for sufficiently fast switching whenever corresponding time-averaged system has a stable synchronization solution.    
     \end{proof} 
     Theorem \ref{th1} reflects that the time-averaged Laplacians $\bar{\mathscr{L}}^{(m)}$, not the frozen time Laplacians $\mathscr{L}^{(m)}(t)$, are the indicator for synchronization in time-varying simplicial complex for sufficient fast rewiring. Hence, the dynamics of the transverse modes given by \eqref{avg_transverse} is the required master stability equation for the complete synchronization in the temporal simplicial complex. Since the time-averaged Laplacian $\bar{\mathscr{L}}^{(1)}$ is orthogonally diagonalizable by the basis of eigenvectors $V^{(1)}$, the master stability equation can be rewritten as,  
     \begin{widetext}
     \begin{equation}\label{average_trans2}
     	\begin{array}{l}
     		\dot{\bar{\zeta}}_{{\perp}_{i}}= [Jf({\bf x}_{0})-\epsilon_{1}{\lambda}^{(1)}_{i} \mathcal{J}H^{(1)}]{\bar{\zeta}}_{{\perp}_{i}}-\sum_{m=2}^{P}\epsilon_{m}\sum_{j=2}^{N}\bar{W}^{(m)}_{{2}_{ij}} \mathcal{J}H^{(m)} \bar{\zeta}_{\perp_{j}}, \; i=2,3,\cdots,N.
     	\end{array}
     \end{equation}
     \end{widetext}
     The stability problem of the synchronization solution is then reduced to solving the master stability equation \eqref{average_trans2} along with the nonlinear equation corresponding to synchronization manifold for the estimation of largest Lyapunov exponent. The necessary condition for the synchronization stability is that $\Lambda_{max}$, largest among all the Lyapunov exponents across the synchronization manifold must be negative. In general, $\Lambda_{max}$ is a function of the eigenvalues of time-averaged Laplacian $\bar{\mathscr{L}}^{(1)}$, the time-averaged generalized Laplacians associated with all the higher-order interactions and the coupling strengths, i.e., 
     \begin{equation*}
     	\begin{array}{l}
     		\Lambda_{max} \equiv \Lambda_{max}(\lambda^{(1)}_{1},\lambda^{(1)}_{2},\cdots,\lambda^{(1)}_{N},\bar{\mathscr{L}}^{(2)},\bar{\mathscr{L}}^{(3)},\cdots,\bar{\mathscr{L}}^{(P)},\\ ~~~~~~~~~~~~~~~~~~~~~~~~~~~~~~~~~~~~~ `\epsilon_{1},\epsilon_{2},\cdots,\epsilon_{P}).
     	\end{array}
     \end{equation*}
     \par At this stage, it is important to notice that the evolution of the transverse modes are not independent from each other, further they are coupled through the coefficients $\bar{W}^{(m)}_{{2}_{ij}}$, obtained from \eqref{avg_triangular_2}. This results in $(N-1)D$-dimensional coupled linear master stability equation \eqref{average_trans2} and for which the calculation of Lyapunov exponents become too costly. In general, it is not possible to decouple this master stability equation into lower dimensional equations as the presence of higher-order interactions yield generalized Laplacian matrices which are not simultaneously diagonalizable. Hence, in the most general framework, the perturbation modes transverse to the synchronization manifold are interlinked with each other, and the stability of the synchronous solution has to be inspected without dimensionality reduction of the master stability equation. However, there are relevant instances for which the variational equation \eqref{average_trans2} can be simplified into lower dimensional form, resulting into complete separation of the transverse modes and reducing the dimension equal to the dimension of unitary dynamical unit. The next subsections will illustrate these results in detail.           
     \subsubsection{Dimensionality reduction by means of coupling functions} \label{natural_coupling} 
     To fully decouple the master stability equation, we first consider a specific class of synchronization noninvasive coupling functions. In particular, the functional forms of the coupling schemes are assumed to be generalized diffusive, i.e., $H^{(1)}({\bf x}_{i},{\bf x}_{j})=g^{(1)}({\bf x}_{j})-g^{(1)}({\bf x}_{i})$, $H^{(2)}({\bf x}_{i},{\bf x}_{j},{\bf x}_{k})=g^{(2)}({\bf x}_{j},{\bf x}_{k})-g^{(2)}({\bf x}_{i},{\bf x}_{i})$ and similarly for $P$-body interactions $H^{(P)}({\bf x}_{i},{\bf x}_{j_1},{\bf x}_{j_2},\cdots,{\bf x}_{j_P})=g^{(P)}({\bf x}_{j_1},{\bf x}_{j_2},\cdots,{\bf x}_{j_P})-g^{(P)}({\bf x}_{i},{\bf x}_{i},\cdots,{\bf x}_{i})$, where $g^{(P)}: \mathbb{R}^{PD} \to \mathbb{R}^{D}$. Moreover, we assume the coupling functions satisfy an additional condition given by,
     \begin{equation}\label{natural}
     	\begin{array}{l}
     		g^{(P)}({\bf x},{\bf x},\cdots,{\bf x})= \cdots=g^{(2)}({\bf x},{\bf x})= g^{(1)}({\bf x})
     	\end{array}
     \end{equation}
     Condition \eqref{natural}, for example, represents the fact that a triadic interaction where two nodes are in same state is identical to a dyadic interaction.   
     \par Now, using condition \eqref{natural}, one can deduce that for each  $m=2,3,\cdots,P$,
     \begin{widetext}
     \begin{equation}\label{natural2}
     	\begin{array}{l}
     		g^{(m)}_{{\bf x}_{j_1}}({\bf x}_{0},{\bf x}_{0},\cdots,{\bf x}_{0}) + g^{(m)}_{{\bf x}_{j_2}}({\bf x}_{0},{\bf x}_{0},\cdots,{\bf x}_{0}) +\cdots+g^{(m)}_{{\bf x}_{j_m}}({\bf x}_{0},{\bf x}_{0},\cdots,{\bf x}_{0})= g^{(1)}_{{\bf x}_{j}}({\bf x}_{0}).
     	\end{array}
     \end{equation}
     \end{widetext} 
     This immediately implies that $\mathcal{J}H^{(P)}=\cdots=\mathcal{J}H^{(2)}= \mathcal{J}H^{(1)}$. As a consequence, we can rewrite the variational equation \eqref{avg_dsimp_linear} in this scenario as,
     \begin{equation}\label{natural_stability_1}
     	\begin{array}{l}
     		\delta \dot{\bar{{\bf X}}}=[I_{N} \otimes Jf({\bf x}_{0})-(\sum_{m=1}^{P}\epsilon_{m}{\bar{\mathscr{L}}}^{(m)}) \otimes \mathcal{J}H^{(1)}] \delta \bar{{\bf X}}.
     	\end{array}
     \end{equation}
     Let us now introduce an effective time-averaged Laplacian matrix $\bar{\mathscr{G}}$, given by  $\bar{\mathscr{G}}=\bar{\mathscr{L}}^{(1)}+r_2\bar{\mathscr{L}}^{(2)}+r_3\bar{\mathscr{L}}^{(3)}+\cdots+r_P\bar{\mathscr{L}}^{(P)}$, where $r_m=\frac{\epsilon_{m}}{\epsilon_{1}}$, $(m=2,3,\cdots,P)$. Since, all the time-averaged Laplacian matrices inherit the symmetric zero row-sum property, the time-averaged matrix $\bar{\mathscr{G}}$ also satisfies this property. The variational equation in terms of the effective matrix $\bar{\mathscr{G}}$ can be represented as  
     \begin{equation}\label{natural_stability_2}
     	\begin{array}{l}
     		\delta \dot{\bar{{\bf X}}}=[I_{N} \otimes Jf({\bf x}_{0})-\epsilon_{1}{\bar{\mathscr{G}}} \otimes \mathcal{J}H^{(1)}] \delta \bar{{\bf X}}.
     	\end{array}
     \end{equation}
     Thereafter, projecting the stake variable $\delta \bar{{\bf X}}$ onto the basis of eigenvectors $V$ corresponding to the effective time-averaged Laplacian $\bar{\mathscr{G}}$ and introducing new variable $\eta=(V\otimes I_{D})^{-1}\delta \bar{{\bf X}}$, we get 
     \begin{equation}\label{natural_stability_3}
     	\begin{array}{l}
     		\dot{\eta}=[I_{N} \otimes Jf({\bf x}_{0})-\epsilon_{1}{\bar{{\Gamma}}} \otimes \mathcal{J}H^{(1)}]{\eta},
     	\end{array}
     \end{equation}
     where $V^{{-1}}\bar{\mathscr{G}}V={\bar{{\Gamma}}}=diag\{0=\gamma_{1}, \gamma_{2}, \gamma_{3}, \cdots, \gamma_{N}\}$ is a diagonal matrix with the eigenvalues $\{0=\gamma_{1} < \gamma_{2} \le \gamma_{3} \le \cdots \le \gamma_{N}\}$ of the effective matrix $\bar{\mathscr{G}}$ in the diagonal entries. Therefore, in this case, the transverse modes fully decouples into $N-1$ numbers of $D$-dimensional linear equations as,    
     \begin{equation}\label{natural_stability_4}
     	\begin{array}{l}
     		\dot{\eta_{i}}=[Jf({\bf x}_{0})-\epsilon_{1}{{{\gamma}}}_{i}\mathcal{J}H^{(1)}]{\eta}_{i}, \;\; i=2,3,\cdots,N
     	\end{array}
     \end{equation}
     {\it Hence, for the temporal simplicial complex in which unitary components are interconnected through generalized diffusive scheme along with an additional coupling condition \eqref{natural}, the synchronization stability can be achieved by solving the $D$-dimensional $(N-1)$ numbers of uncoupled linear equations $\eqref{natural_stability_4}$ for the calculation of maximum Lyapunov exponent.}       
     \subsubsection{Dimensionality reduction by restricting the topology of connectivity structure} \label{reduction_2}
     Another interesting scenario regarding the dimensionality reduction of the master stability equation occurs when we restrict the topology of the connectivity structure of the time-varying simplicial complex \eqref{gen_model}. We consider the connection topology in such a way that the time-averaged Laplacians $\bar{\mathscr{L}}^{(m)}$, $m=1,2,\cdots,P$ are real symmetric and any one of these Laplacian matrices commute with all the other. Without loss of generality, we assume that the time-averaged Laplacian $\bar{\mathscr{L}}^{(1)}$ which accounts for the pairwise interactions among the unitary elements of the simplicial complex, commutes with all the other generalized time-averaged Laplacians related to higher-order interactions.
     \par As the time-averaged Laplacians $\bar{\mathscr{L}}^{m}$ are real symmetric matrices, they all are diagonalizable by the set of linearly independent eigenvectors $V^{(m)}$, i.e., $V^{(m)^{-1}}\bar{\mathscr{L}}^{(m)}V^{(m)}= diag \{0=\lambda_{1}^{(m)},\lambda_{2}^{(m)},\lambda_{3}^{(m)}, \cdots, \lambda_{N}^{(m)}\}$. Again all the generalized time-averaged Laplacians $\bar{\mathscr{L}}^{(m)}$, $m=2,3,\cdots,P$ are commuting with $\bar{\mathscr{L}}^{(1)}$. So, for each $m=2,3,\cdots,P$ the time-averaged Laplacians $\bar{\mathscr{L}}^{(1)}$ and $\bar{\mathscr{L}}^{(m)}$ share same basis of orthogonal eigenvectors that diagonalize them, i.e., $V^{(1)}=V^{(m)}$. This means that all the time-averaged Laplacians $\bar{\mathscr{L}}^{m}$, $(m=1,2,\cdots,P)$ are orthogonally diagonalizable by the set of eigenvectors $V^{(1)}$, i.e., $V^{(1)^{-1}}\bar{\mathscr{L}}^{(m)}V^{(1)}= diag \{0=\lambda_{1}^{(m)},\lambda_{2}^{(m)},\lambda_{3}^{(m)}, \cdots, \lambda_{N}^{(m)}\}$.
     \par Therefore, in this case, the master stability equation \eqref{average_trans2} can be rewritten as  
     \begin{equation}\label{commutative_trans2}
     	\begin{array}{l}
     		\dot{\bar{\zeta}}_{{\perp}_{i}}= [Jf({\bf x}_{0})-\sum_{m=1}^{P}\epsilon_{m}{\lambda}_{i}^{(m)}\mathcal{J}H^{(m)}] \bar{\zeta}_{{\perp}_{i}}, \\\\~~~~~~~~~~~~~~~~~~~~~~~~~~~~~~~~~~~~~~~~~~~~~~~ i=2,3,\cdots,N.
     	\end{array}
     \end{equation}
     Surely it is a $(N-1)$ numbers of $D$-dimensional uncoupled system of equation. {\it Hence for the temporal simplicial complex \eqref{gen_model} with connection topology between unitary elements such that the time-averaged Laplacian associated with pairwise connectivity commutes with all the other time-averaged Laplacians corresponds to higher-order interactions, the master stability equation can be fully decoupled into lower dimensional equations of dimension equals to the dimension of unitary dynamical unit.} 
     \begin{corollary}\label{col1}
     	Let us consider $P=2$, i.e., only two and three-body interactions are allowed among the unitary components of the simplicial complex \eqref{gen_model} and further the time-averaged Laplacian $\bar{\mathscr{L}}^{(1)}$ related to pairwise interactions have one zero eigenvalue and $N-1$ numbers of equal non-zero eigenvalue. Without loss of generality, we assume $\lambda_{1}^{(1)}=0$ and $\lambda_{i}^{(1)}=\lambda$, for $i=2,3,\cdots,N$. Then the master stability equation \eqref{average_trans2} can be decoupled into $N-1$ numbers of linear $D$-dimensional equations.  
     \end{corollary}
     \begin{proof}[Proof of corollary \ref{col1}]
     	Since in this case, $\lambda_{1}^{(1)}=0$ and $\lambda_{i}^{(1)}=\lambda$, for $i=2,3,\cdots,N$. Therefore, the master stability equation can be rewritten as
     	\begin{equation}\label{p=2_trans}
     		\begin{array}{l}
     			\dot{\bar{\zeta}}_{{\perp}_{i}}= [Jf({\bf x}_{0})-\epsilon_{1}{\lambda} \mathcal{J}H^{(1)}]{\bar{\zeta}}_{{\perp}_{i}}-\epsilon_{2}\sum_{j=2}^{N}\bar{W}^{(2)}_{{2}_{ij}} \mathcal{J}H^{(2)} \bar{\zeta}_{\perp_{j}}, \\\\~~~~~~~~~~~~~~~~~~~~~~~~~~~~~~~~~~~~~~~~~~~~~~~ i=2,3,\cdots,N.
     		\end{array}
     	\end{equation}
     	Now $\bar{\mathscr{L}}^{(2)}$ is a real symmetric matrix and orthogonally diagonalizable by its set of eigenvector $V^{(2)}$. Further, $\bar{W}_{2}^{(2)}$ satisfies the relation
     	\begin{equation}\label{p=2_avg_triangular}
     		\begin{array}{l}
     			V^{(1)^{-1}}\bar{\mathscr{L}}^{(2)}V^{(1)}= \bar{W}^{(2)}= 
     			\begin{bmatrix}
     				0 & 0_{1\times (N-1)} \\
     				0_{(N-1)\times 1} & \bar{W}^{(2)}_{2}
     			\end{bmatrix},
     		\end{array}
     	\end{equation}
     	where $V^{(1)}$ is the basis of orthogonal eigenvectors of $\bar{\mathscr{L}}^{(1)}$. Hence using the symmetric property of $\bar{\mathscr{L}}^{(2)}$, orthogonal property of the eigenvectors and the relation \eqref{p=2_avg_triangular}, one can conclude that $\bar{W}^{(2)}_{2}$ is also real symmetric matrix of order $N-1$. Consequently, it is orthogonally diagonalizable by its set of eigenvectors $V'$, i.e., $(V')^{-1}\bar{W}^{(2)}_{2}V'= diag\{\alpha_{1},\alpha_{2},\cdots,\alpha_{N-1}\}$, where $\alpha_{i}$, $(i=1,2,\cdots,N-1)$ are the eigenvalues of $\bar{W}^{(2)}_{2}$. 
     	\par Now, projecting the transverse components $\bar{{\zeta}}_{\perp_i}$ to the basis of eigenvector $V'$ and introducing new variable $\phi=(V'\otimes I_{D})^{-1}\bar{{\zeta}}_{\perp}$, the variational equation \eqref{p=2_trans} becomes
     	\begin{equation}\label{p=2_trans2}
     		\begin{array}{l}
     			\dot{\phi}_{{i}}= [Jf({\bf x}_{0})-\epsilon_{1}{\lambda} \mathcal{J}H^{(1)}-\epsilon_{2}\alpha_{i} \mathcal{J}H^{(2)}] \phi_{{i}}, \\\\ ~~~~~~~~~~~~~~~~~~~~~~~~~~~~~~~~~~~~~~~~~ i=1,2,\cdots,N-1.
     		\end{array}
     	\end{equation}
     	This is the required fully decoupled master stability equation.     
     \end{proof}
     \begin{remark}
     	Corollary \ref{col1} reflects the fact that one can obtain the fully decoupled form of the master stability equation without the time-averaged Laplacians being commutative. For example, we consider the connection topology of the pairwise interactions to be random network with probability of generating link between any two nodes as $p_{rand}$. Then the corresponding time-averaged Laplacian $\bar{\mathscr{L}}^{(1)}$ is, 
     	\[ \bar{\mathscr{L}}^{(1)}_{ij}= \begin{cases}
     		-p_{rand}, & i\ne j \\
     		(N-1)p_{rand} & i=j.
     	\end{cases}\]
     	Clearly the eigenvalues of $\bar{\mathscr{L}}^{(1)}$ are $\lambda_{1}^{(1)}=0$ and $\lambda_{i}^{(1)}=Np_{rand}$, for each $i=2,3,\cdots,N$. Hence, in this scenario the master stability equation fully decouples by the corollary \ref{col1}.
     \end{remark}
     \section{Numerical results} \label{numerical} 
     In this section, we provide a series of numerical findings to demonstrate our theoretical finding and confirm the applicability of our findings in wide range of scenarios. Particularly, we consider temporal simplicial complexes whose individual node dynamics are given by three different $3$-dimensional $(\mathbf{x}=[x,y,z]^{tr})$ chaotic systems, namely as a real-world example of neuronal evolution, on Sherman model \cite{sherman1994anti,jalil2012spikes} of pancreatic $\beta$-cell, and two prototypical chaotic systems, the Lorenz \cite{strogatz2018nonlinear} and R\"{o}ssler \cite{rossler1976equation} oscillators. We first consider the most prevalent occurrence on Sherman model and subsequently we illustrate the results for decoupled master stability equation on Lorenz and R\"{o}ssler systems. We further continue our analysis on a real-world temporal network structure.
     \par For the numerical simulations, we consider the temporal simplicial complex of $N=100$ nodes, and the pairwise connections between the nodes are rewired at each time instance $\delta t$ with a probability $r\delta t$, where $r$ is the rewiring frequency. All the higher-order structures (e.g., $2$-simplices, $3$-simplices) are then  constructed at each time instance by considering the cliques of dimension $m$ as $m$-simplices. In this way, all the pairwise links and higher-order connections get rewired stochastically at each time instance $\delta t$ with a rewiring frequency $r$. Sufficiently smaller value of $r$ indicates very slow switching between the connectivity structures of the simplicial complex (almost static), while adequately large $r$ corresponds to fast switching among the topologies of the simplicial complexes over the total time period. For simplicity, we here consider only the higher-order interaction up to interconnections between $3$-nodes. A synchronization measure is introduced as average synchronization error by,
     \begin{equation} \label{error}
     	\begin{array}{l}
     		E=\lim\limits_{\mathcal{T}\to\infty}\dfrac{1}{\mathcal{T}}\bigintss_{t_{0}}^{t_{0}+\mathcal{T}} E(\tau)d\tau ,
     	\end{array} 
     \end{equation}              
     where $E(\tau)=\sum\limits_{i,j=1}^{N} \dfrac{\|\mathbf{x}_{j}(\tau)-\mathbf{x}_{i}(\tau)\|}{N(N-1)}$ is the synchronization error at each time stamp $\tau$, $t_{0}$ is the initial transient and $\mathcal{T}$ is an adequately large number. The onset of complete synchrony in the simplicial complex \eqref{gen_model} is indicated by zero value of the average synchronization error $E$. Here, the $4^{th}$-order Runge-Kutta scheme is used for the execution of numerical integration in the temporal simplicial complex \eqref{gen_model} with integration step $\delta t = 0.001$ (for Sherman model) and $\delta t = 0.01$ (for Lorenz and R\"{o}ssler system), in a period of time iteration $10^{5}$ after discarding initial $4 \times 10^{5}$ time iterations. The initial conditions for the numerical integration are selected randomly from the phase space of unitary node dynamics and all the results are obtained by taking average over a total of $10$ realizations. 
     \subsection{The general scenario} \label{coupled_mse}
     \begin{figure*}[ht] 
     	\centerline{
     		\includegraphics[scale=0.3]{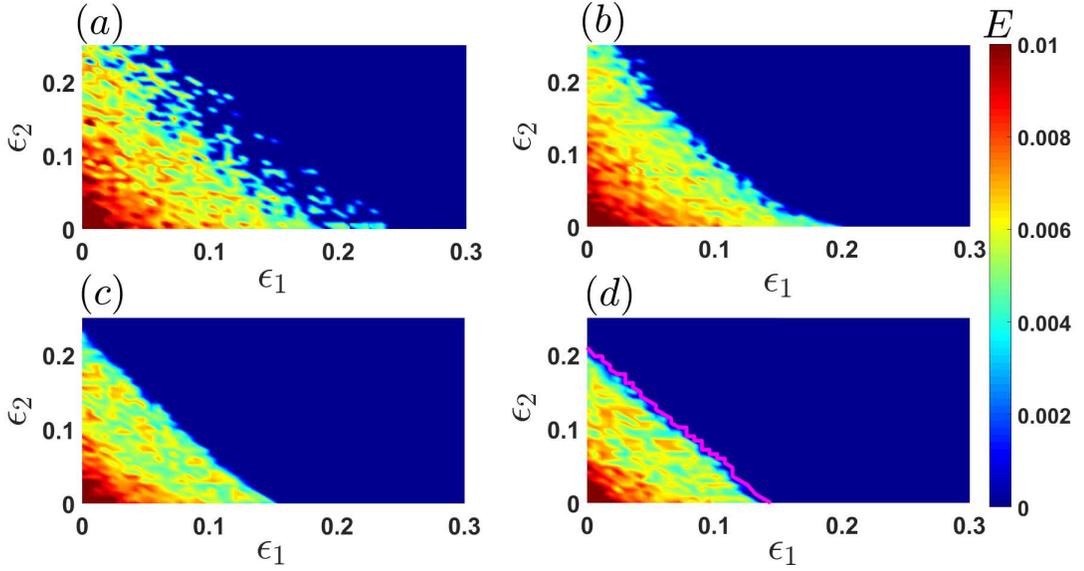}}
     	\caption{{\bf Synchronization in time-varying simplicial complex of coupled Sherman neurons.} Average synchronization error E in the $(\epsilon_{1},\epsilon_{2})$ parameter region for various rewiring frequencies. For (a) $r=10^{-6}$, (b) $r=10^{-3}$, (c) $r=10^{0}$, and (d) $r=10^{3}$, respectively. Vertical colorbar represents the variation of $E$ with deep blue color describing the region of synchronization. The continuous magenta curve superimposed to subfigure (d) indicates the curve of theoretically predicted synchronization thresholds obtained from Eq. \eqref{average_trans2}. The values of other parameters are $p_{sw}=0.15$ and $k_{sw}=4$.}
     	\label{sherman_eps1_eps2}
     \end{figure*} 
     We start our analysis by considering the most general case, where the master stability equation can not be decoupled without restricting conditions on coupling functions or connectivity topology of the simplicial complex. For example, we assume temporal simplicial complex of coupled Sherman model, whose coupling schemes corresponding to pairwise and triadic interactions are given by $H^{(1)}({\bf x}_i,{\bf x}_{j})=[tanh(2(x_j-x_i)),\; 0,\; 0]^{tr}$ and $H^{(2)}({\bf x}_i,{\bf x}_{j},{\bf x}_{k})=[tanh(2(x_j+x_k-2x_i)),\; 0,\; 0]^{tr}$. Then the dynamics of the simplicial complex \eqref{gen_model} is represent by,
     \begin{equation} \label{sherman_eq}
     	\begin{array}{l}
     		\beta\dot{x}_{i}=-{\chi}_{Ca}(x_{i})-\chi_{K}(x_{i},y_{i})-\chi_{S}(x_{i},z_{i})
     		\\ ~~~~~~~~~~~~~~~~~~~~~~~ +\epsilon_{1} \sum\limits_{j=1}^{N}{\mathscr{B}_{ij}^{(1)}}(t)tanh(2(x_j-x_i)) 
     		\\ ~~~~~~~~~~~~~~~~~~~~~~~ +\epsilon_{2}\sum\limits_{j=1}^{N}\sum\limits_{k=1}^{N}\mathscr{B}_{ijk}^{(2)}(t)tanh(2(x_j+x_k-2x_i)), \\\\
     		\beta\dot{y}_{i}=\nu [y^{\infty}(x_{i})-y_{i}], \\\\
     		\beta_{s}\dot{z}_{i}=z^{\infty}(z_{i})-z_{i}, 
     	\end{array}
     \end{equation} 	
     where two fast currents: calcium $\chi_{Ca}$,  persistent potassium $\chi_{K}$, and a slow potassium current $\chi_{S}$ are given by $\chi_{S}(x,z)=G_{S}z(x-E_{K})$, $\chi_{K}(x,y)=G_{K}y(x-E_{K})$, and $\chi_{Ca}(x)=G_{Ca}m^{\infty}(x-E_{Ca})$, respectively.
     $x$ is the membrane potential corresponding to the reversal potential $E_{K}=-0.075$, $E_{Ca}=0.025$. $m$, $y$, and $z$ are the voltage dependent gating variables. The maximum conductance and time constants are $G_{Ca}=3.6$, $G_{K}=10$, $G_{S}=4$ and $\beta=0.02$, $\beta_{s}=5$. $\nu=1$, an auxiliary scaling factor, manages the time scale of the persistent potassium channels. The values of the gating variables at steady state are 
     \begin{equation*}
     	\begin{array}{l}
     		m^{\infty}(x)=\{1+\exp[-83.34(x+0.02)]\}^{-1}, \\\\y^{\infty}(x)=\{1+\exp[-178.57(x+0.016)]\}^{-1}, \;\mbox{and}\\\\ z^{\infty}(x)=\{1+\exp[-100(x+0.035345)]\}^{-1}.
     	\end{array}	
     \end{equation*}
     The parameter values are chosen in a manner that results in chaotic bursting behavior in the individual node dynamics of the simplicial complex.
     \par At each time instance $t$, we consider that the topology of the connectivity structure for the pairwise interactions $\mathscr{B}^{(1)}(t)$ are obtained from Watts-Strogatz small-world network algorithm \cite{watts1998collective}, i.e., staring from a ring of $N$ numbers of nodes with $k_{sw}$ nearest neighbors on each side and the edges of the ring are rewired with a certain probability $p_{sw}$. Consequently, the adjacency tensors $\mathscr{B}^{(2)}(t)$ for three-body interactions are obtained by considering the $m$-dimensional cliques as $m$-simplices. All these adjacency matrices and tensors are then rewired in time with a rewiring frequency $r$.    
     \par Figure \ref{sherman_eps1_eps2} delineates the result corresponding to complete synchrony in $(\epsilon_{1},\epsilon_{2})$ parameter plane for different values of rewiring frequencies, where the colorbar indicates the variation of synchronization error $E$ (deep blue color indicates the occurrence of synchronization). The region of synchronization by simultaneously varying $\epsilon_{1} \in [0, 0.3]$ and $\epsilon_{2} \in [0, 0.25]$ for sufficiently slow switching $(r=10^{-6})$ is depicted in Fig. \ref{sherman_eps1_eps2}(a). As observed, synchronization can be achieved for sole presence of pairwise interaction, but the presence of only three-body interaction is not sufficient for the emergence of synchrony. However, as the non-pairwise coupling strength $\epsilon_{2}$ increases, the critical value of the pairwise coupling strength $\epsilon_{1}$ for the emergence of synchronization decreases. We then increase the rewiring frequency to $r=10^{-3}$. and plot a similar parameter space in Fig. \ref{sherman_eps1_eps2}(b). In this situation, the structure of the simplicial complex varies in time more frequently than the previous, where the simplicial complex was almost static in time. As a result synchronization emerges for lower values of coupling strengths. Still the presence of only three-body interaction is not sufficient to achieve synchrony. For much higher rewiring frequencies $r=10^{0}$ (Fig. \ref{sherman_eps1_eps2}(c)) and $r=10^{3}$ (Fig. \ref{sherman_eps1_eps2}(d)), the enhancement in region of synchrony is more prominent, as the connection topology changes much faster in time. Surprisingly, for this higher frequencies, synchronization is achievable using either the pairwise or higher-order interactions only.
     \par To verify that our theoretical findings are in good agreement with the numerical results, we solve the coupled master stability equation \eqref{average_trans2} for the estimation of maximum transverse Lyapunov exponent $\Lambda_{max}$. The solid magenta line superimposed to Fig. \ref{sherman_eps1_eps2}(d), represents the critical stability curve characterized by the zero value of $\Lambda_{max}$. The region below and above (i.e., positive and negative value of maximum Lyapunov exponent) this stability curve indicate the region of synchronization and desynchronization in the $(\epsilon_{1},\epsilon_{2})$ parameter plane. It is clearly observable that the theoretical predictions obtained through time-averaged structure are in complete agreement with the numerical results acquired for sufficiently fast switching $(r=10^{3})$ of the connection topology.
     \begin{figure*}[ht] 
     	\centerline{
     		\includegraphics[scale=0.3]{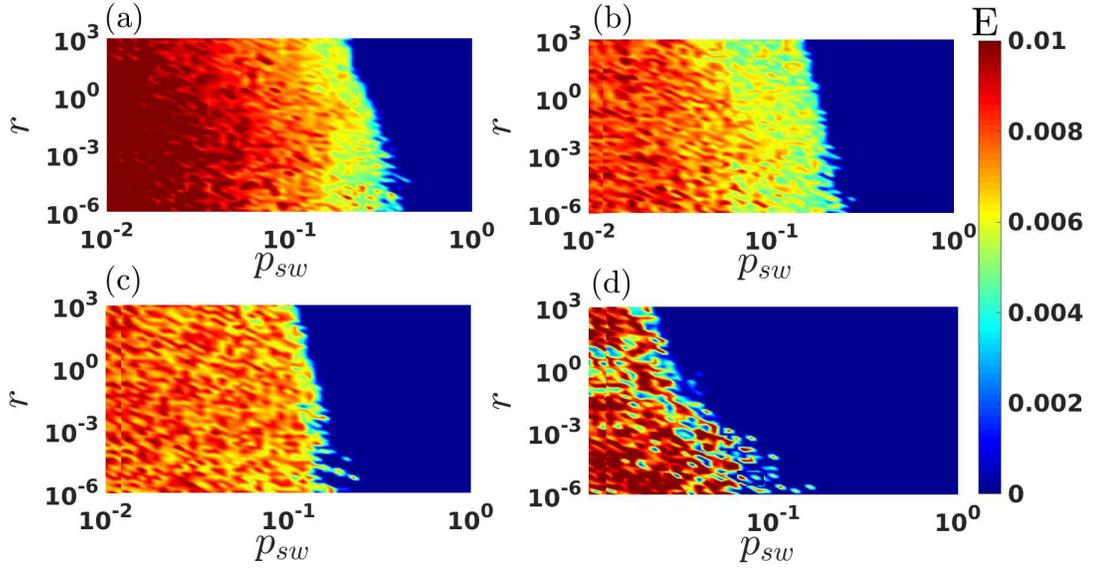}}
     	\caption{{\bf Synchronization in time-varying simplicial complex of coupled Sherman neurons.} Average synchronization error E in the $(p_{sw},r)$ parameter region for different three-body coupling strengths with fixed pairwise coupling strength. For (a) $\epsilon_{2}=0$, (b) $\epsilon_{2}=0.05$, (c) $\epsilon_{2}=0.1$, and (d) $\epsilon_{2}=0.2$, respectively. Vertical colorbar represents the variation of $E$ with deep blue color describing the region of synchronization. Other parameter values are $\epsilon_{1}=0.1$ and $k_{sw}=4$.}
     	\label{sherman_psw_r}
     \end{figure*}
     \par Next we investigate the effect of small-world probability $p_{sw}$ on the emergence of synchronization for varying rewiring frequency $r$. The parameter $p_{sw}$ serves as one of the most crucial variable in this investigation since it establishes the randomness that exists in the connectivity topology of the simplicial complex at any particular time instance. Figure \ref{sherman_psw_r} elucidates the variation of average synchronization error in $(p_{sw},r)$ parameter plane for different values of triadic coupling strength $\epsilon_{2}$ along with fixed pairwise coupling strength $\epsilon_{1}=0.1$. Figure \ref{sherman_psw_r}(a) shows the result corresponding to the instance when the three-body interactions are not considered (i.e., $\epsilon_{2}=0$). We observe that for sufficiently large small-world probability $(p_{sw} \gtrapprox 4.78\times 10^{-1})$, synchronization is achievable regardless of slow or fast rewiring frequencies and the critical values of small-world probability for the emergence of synchronization decreases with increasing rewiring frequency. But, for smaller probability $(p_{sw}\lessapprox 2.1\times 10^{-1})$, no synchronization is plausible despite how quickly the rewiring happens. A qualitatively same behavior can be observed as the three-body interactions are taken into consideration with non-pairwise coupling strengths $\epsilon_{2}=0.05$ (Fig. \ref{sherman_psw_r} (b)) and $\epsilon_{2}=0.1$ (Fig. \ref{sherman_psw_r} (c)). In these cases, just the region of synchronization enhances due to the simultaneous presence of both pairwise and three-body interactions. However, for strong higher-order coupling strength $\epsilon_{2}=0.2$ (Fig. \ref{sherman_psw_r}(d)), a quite interesting result is observed. In this scenario, the synchronization is achievable even with very low small-world probability $(p_{sw} \approx 2.3\times 10^{-2})$ for  sufficiently fast rewiring of connectivity structure. Therefore, one can conclude that with strong higher-order coupling and adequately fast rewiring, the temporal simplicial complex achieves synchronization even for very low edge-rewiring probability $p_{sw}$ of the small-world connection topology. That is, when the connections between individuals of a simplicial complex are constructed following small-world algorithm, comparably little randomness in the connection topology is sufficient to achieve a synchronization state with quick rewiring and strong non-pairwise coupling.           
     \subsection{Lower dimensional cases}         
     We now consider an example to verify our analytical approach that gives fully decoupled lower dimensional master stability equation. We, therefore, start with the instance where the nodes of the simplicial complex are interconnected by means of coupling schemes as discussed in Sec. \ref{natural_coupling}. In particular, a temporal simplicial complex of coupled R\"{o}ssler oscillators is considered whose coupling schemes related to the pairwise and three-body interactions satisfy the natural coupling condition \eqref{natural}. For example, we choose the coupling functions as $H^{(1)}({\bf x}_i,{\bf x}_{j})=[0,\; y_{j}^{3}-y_{i}^{3},\; 0]^{tr}$ and $H^{(2)}({\bf x}_i,{\bf x}_{j},{\bf x}_{k})=[0,\;y_{j}^{2}y_{k}-y_{i}^{3} ,\; 0]^{tr}$. Then the corresponding equation of motion of the temporal simplicial complex is given by, 
     \begin{equation}\label{rossler}
     	\begin{array}{l}
     		\dot{x}_{i}= -y_{i} -z_{i}, \\\\
     		\dot{y}_{i}= x_{i} + a y_{i}+ \epsilon_{1} \sum_{j=1}^{N} \mathscr{B}^{(1)}_{ij}(t) (y_{j}^{3}-y_{i}^{3}) 
     		\\ ~~~~~~~~~~~~~~~~ + \epsilon_{2} \sum_{j=1}^{N}\sum_{k=1}^{N} \mathscr{B}^{(2)}_{ijk}(t) (y_{j}^{2}y_{k}-y_{i}^{3}), \\\\
     		\dot{z}_{i}= b + z_{i} (x_{i}-c), ~~~~~~~ i=1,2,\cdots,N,        
     	\end{array}
     \end{equation}  
     where the system parameters are taken as $a=0.2$, $b=0.2$ and $c=5.7$, for which the unitary nodes display chaotic dynamics. Here also at each time instance the connection between the nodes are established following the Watts-Strogatz small-world graph model.
     \begin{figure*}[ht] 
     	\centerline{
     		\includegraphics[scale=0.3]{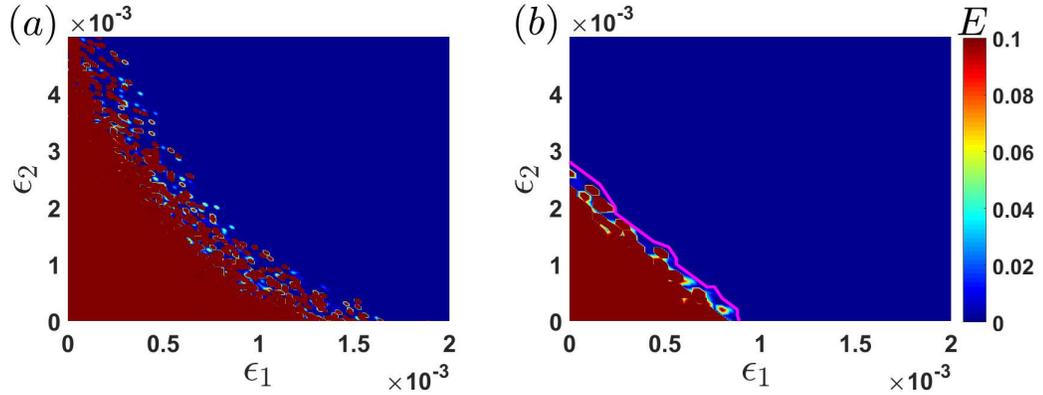}}
     	\caption{ {\bf Synchronization in time-varying simplicial complex of coupled R\"{o}ssler oscillators.} Average synchronization error E in the $(\epsilon_{1},\epsilon_{2})$ parameter region for slow and fast rewiring frequencies. For (a) $r=10^{-4}$, and (b) $r=10^{2}$. Vertical colorbar represents the variation of $E$ with deep blue describing the region of synchronization. The solid magenta line overlaid in panel (b) describes the curve of theoretically predicted synchronization threshold obtained from Eq. \eqref{natural_stability_4}. The values of all the other parameters are $p_{sw}=0.15$ and $k_{sw}=4$. }
     	\label{rossler_eps1_eps2}
     \end{figure*}
     Figure \ref{rossler_eps1_eps2} portrays the region of synchronization and de-synchronization by evaluating average synchronization error $E$ as a function of pairwise and higher-order coupling strengths. The result for almost static simplicial complex $(r=10^{-4})$ is displayed in Fig. \ref{rossler_eps1_eps2}(a). One can observe that in the absence of three-body interactions (i.e., $\epsilon_{2}=0$), the temporal simplicial complex can achieve synchronization state with only pairwise interactions and as the value of $\epsilon_{2}$ increases the critical value of $\epsilon_{1}$ for the emergence of synchronization decreases. However, only the non-pairwise interactions are not sufficient for the achievement of synchrony. On the other hand, when the structure of the simplicial complex changes rapidly in time (i.e., $r=10^{2}$), sole presence of both dyadic and triadic interactions is sufficient for synchronization [Fig. \ref{rossler_eps1_eps2}(b)]. Further in this situation, a larger negative slope guarantees that the region of synchronization (deep blue region) enhances by a huge margin compared to slow switching of the connections. In Fig. \ref{rossler_eps1_eps2}(b), the continuous magenta curve overlaid with the $(\epsilon_{1},\epsilon_{2})$ parameter space indicates the curve corresponding to null value of maximum Lyapunov exponent $\Lambda_{max}$, obtained by solving the analytically derived master stability equation \eqref{natural_stability_4}. The regions below and above this critical curve represent the regions of desynchrony and synchrony, respectively. As perceived, the analytical curve perfectly matches with the numerical obtained critical coupling values for the emergence of synchronization. Hence, in case of this specific coupling scheme, our analytically derived fully decoupled master stability equation \eqref{natural_stability_4} is able to accurately predict the critical coupling strengths for the emergence of synchronization in temporal simplicial complexes.
     \par Next we move to the results of the scenario, where some restrictions on the connection topology of the temporal simplicial complex give rise to fully decoupled master stability equation. We consider that at each time instance the nodes of the simplicial complex are connected through Erd\H{o}s-R\'enyi random network algorithm \cite{erdos2011evolution}, i.e., any two nodes are connected through an edge with a probability $p_{rand}$ and the connection between group of $3$-nodes are established by promoting each cliques of dimension $2$ to $2$-simplices. The Lorenz system is considered as the dynamics of each individual node and therefore, the evolution of the $i$-th node of the temporal simplicial complex is given by the following set of equations,
     \begin{equation}\label{lorenz}
     	\begin{array}{l}
     		\dot{x}_{i}= \sigma(-y_{i} -x_{i})+ \epsilon_{2} \sum_{j=1}^{N}\sum_{k=1}^{N} \mathscr{B}^{(2)}_{ijk}(t) (x_{j}^{2}x_{k}-x_{i}^{3}), \\\\
     		\dot{y}_{i}= \rho x_{i} - x_{i} z_{i}-y_{i}+ \epsilon_{1} \sum_{j=1}^{N} \mathscr{B}^{(1)}_{ij}(t) (y_{j}-y_{i}), \\\\
     		\dot{z}_{i}= x_{i}y_{i}-\beta z_{i}, ~~~~~~~ i=1,2,\cdots,N,        
     	\end{array}
     \end{equation}      
     where the coupling forms corresponding to two and three-body interactions are taken as $H^{(1)}({\bf x}_i,{\bf x}_{j})=[0,\; y_{j}-y_{i},\; 0]^{tr}$ and $H^{(2)}({\bf x}_i,{\bf x}_{j},{\bf x}_{k})=[x_{j}^{2}x_{k}-x_{i}^{3},\; 0,\; 0]^{tr}$, respectively, and the system parameters are considered to be $\sigma=10$, $\beta=\frac{8}{3}$, $\rho=28$ for the occurrence of chaotic dynamics.
     \begin{figure*}[ht] 
     	\centerline{
     		\includegraphics[scale=0.3]{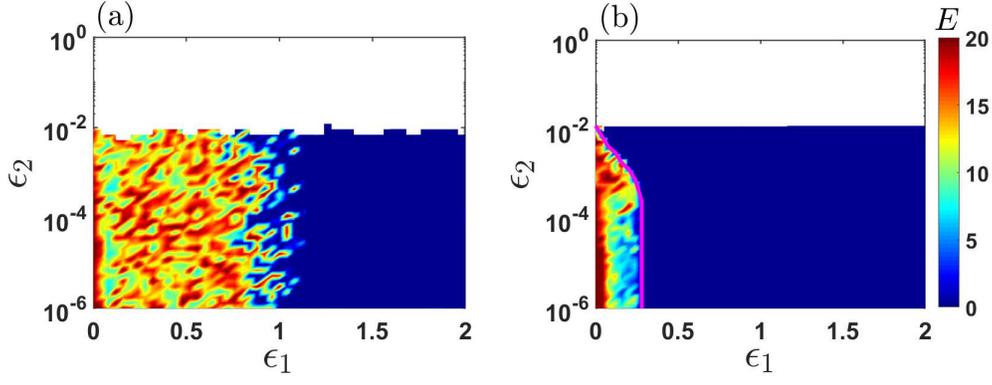}}
     	\caption{{\bf Synchronization in time-varying simplicial complex of coupled Lorenz systems.} Average synchronization error E in the $(\epsilon_{1},\epsilon_{2})$ parameter region for almost static and sufficiently fast rewiring. For (a) $r=10^{-4}$, and (b) $r=10^{2}$. Vertical colorbar represents the variation of $E$ with deep blue region describing the stable synchronization state. The magenta curve superimposed to panel (b) indicates the curve of theoretically predicted synchronization thresholds obtained from Eq. \eqref{p=2_trans2}. The white areas represent the unbounded region. Here the probability of generating an edge between any two node is $p_{rand}=0.1$.}
     	\label{lorenz_eps1_eps2}
     \end{figure*}
     \par Figure \ref{lorenz_eps1_eps2} portrays the average synchronization error $E$ as a function of $\epsilon_{1}$ and $\epsilon_{2}$. For slow rewiring frequency $(r=10^{-4})$, up to $\epsilon_{1} \approx 1$, the synchronization in not achievable with any value of higher-order coupling strength $\epsilon_{2}$ [see fig. \ref{lorenz_eps1_eps2} (a)]. Further increment of pairwise coupling strength $\epsilon_{1}$ results in the occurrence of synchrony. But, as soon as $\epsilon_{2}$ crosses the value $10^{-2}$, synchronization is forbidden for every value of $\epsilon_{1}$, as the system becomes unbounded (white region). On the other hand, for very fast rewiring frequency $(r=10^{-4})$, synchronization is achievable just beyond $\epsilon_{1} \approx 0.28$ [cf. Fig. \ref{lorenz_eps1_eps2}(b)]. Consequently, an enormous improvement in the region of synchrony is observed. Furthermore, in Fig. \ref{lorenz_eps1_eps2}(b), the magenta black line laid over the plot of average synchronization error, represents the coupling thresholds for the emergence of synchronization, obtained by solving the decoupled master stability equation (Eq. \eqref{p=2_trans2}) for maximum Lyapunov exponent $\Lambda_{max}$. As observed, the theoretically predicted synchronization threshold curve is in good agreement with the numerical result corresponding to rapid switching of the connectivity structure.
     \begin{figure*}[ht] 
     	\centerline{
     		\includegraphics[scale=0.3]{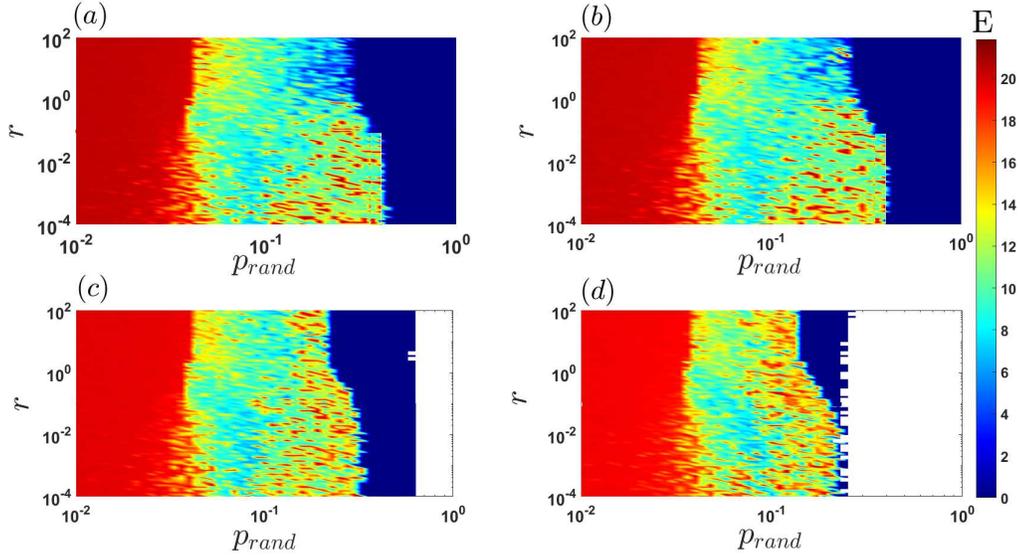}}
     	\caption{{\bf Synchronization in time-varying simplicial complex of coupled Lorenz systems.} Average synchronization error E in the $(p_{rand},r)$ parameter region for various triadic coupling strengths with fixed pairwise coupling strength $\epsilon_{1}=0.1$. For (a) $\epsilon_{2}=0$, (b) $\epsilon_{2}=10^{-5}$, (c) $\epsilon_{2}=10^{-4}$, and (d) $\epsilon_{2}=10^{-3}$. Vertical colorbar represents the variation of average synchronization error $E$ with deep blue domain describing the region of synchronization. The white areas represent the unbounded regions.}
     	\label{lorenz_prand_r}
     \end{figure*}
     \subsubsection{Effect of higher-order coupling and temporal rewiring on emergence of synchrony in low connectivity regime} An essential condition for complete synchronization is the connectivity of underline network topology, hence it can only be achieved when the network structure is suitably connected. However, the synchronization state is achievable in many real-world temporal systems despite their instantaneous sparse connectivity between individuals. Therefore, a natural question arises about the mechanism for the emergence of synchrony in such low connectivity regime. To obtain a plausible strategy, we scrutinize the combined effect of higher-order interaction and temporal switching on the emergence of synchronization in varying connectivity regime of the simplicial complex. Figure \ref{lorenz_prand_r} delineates $E(p_{rand},r)$ for different values of non-pairwise coupling strength $\epsilon_{2}$ with fixed pairwise coupling $\epsilon_{1}=0.1$. As observed, in the sole presence of dyadic interactions (i.e., $\epsilon_{2}=0$ ), the system always achieve synchronization for large values of connectivity probability ($p_{rand} \gtrapprox 4.38 \times 10^{-1}$) regardless of slow or fast rewiring of the topology [see Fig. \ref{lorenz_prand_r}(a)]. In this connectivity regime, synchronization is in fact expected due to the presence of sufficient number of connections between the system individuals. However, the critical connectivity probability for the emergence of synchrony decreases as the rewiring frequency increases. But, for the lower values of $p_{rand}$ ($\lessapprox 3 \times 10^{-1}$), synchronization is forbidden no matter how fast the topologies are rewired. We then introduce the triadic connections $(\epsilon_{2} > 0)$ along with the pairwise ones, and corresponding results are reported in Figs. \ref{lorenz_prand_r}(b)- \ref{lorenz_prand_r}(d). The presence of three-body interactions requires qualitatively smaller critical connection probability $p_{rand}$ for the achievement of synchronization state both for slow and fast switching. For sufficiently fast rewiring $(r=10^{2})$, as the triadic coupling is increased just to $\epsilon_{2} = 10^{-3}$, the threshold to achieve synchrony decreases to $p_{rand} \approx 1.41 \times 10^{-1}$ [Fig. \ref{lorenz_prand_r}(d)]. Nevertheless, the presence of sufficiently larger triadic coupling strength and connectivity probability jointly makes the system unbounded, which is reflected in Figs. \ref{lorenz_prand_r}(c) and \ref{lorenz_prand_r}(d) (white region). Further increment of non-pairwise coupling to $\epsilon_{2}=10^{-2}$ makes the system unbounded beyond $p_{rand} = 10^{-1}$ for almost static simplicial complex $(r=10^{-4})$ without achieving desynchrony to synchrony transition (cf. Fig. \ref{lorenz_prand_r_fast}). In contrary, for sufficiently fast switching, the system goes through desynchrony to synchrony transition and the critical value for the achievement of synchronization state is considerably smaller ($p_{rand} \approx 7.5 \times 10^{-2}$). {\it Therefore, for the simplicial complex with random connections between the individuals, synchronization state is achievable even in very low connectivity regime when the higher-order coupling is sufficiently large and the rewiring among the structures of simplicial complex happens quickly.} The obtained result can be interpreted as a fundamental mechanism for emergence of synchrony in neuronal dynamics and formation of mutual agreement in virtual communities. Notably, our finding indicates that groups of people can unify their perspectives amid relatively sparse exchange of viewpoints, and  neurons can synchronize in spite of limited connections among them, which are clear results of speedy communications that occur in contemporary social sites and fast signal transmissions in neuronal populations. This result may also be useful in describing the emergence of synchronization in power grid networks with limited communication connectivity.       
     \begin{figure}[ht] 
     	\centerline{
     		\includegraphics[scale=0.45]{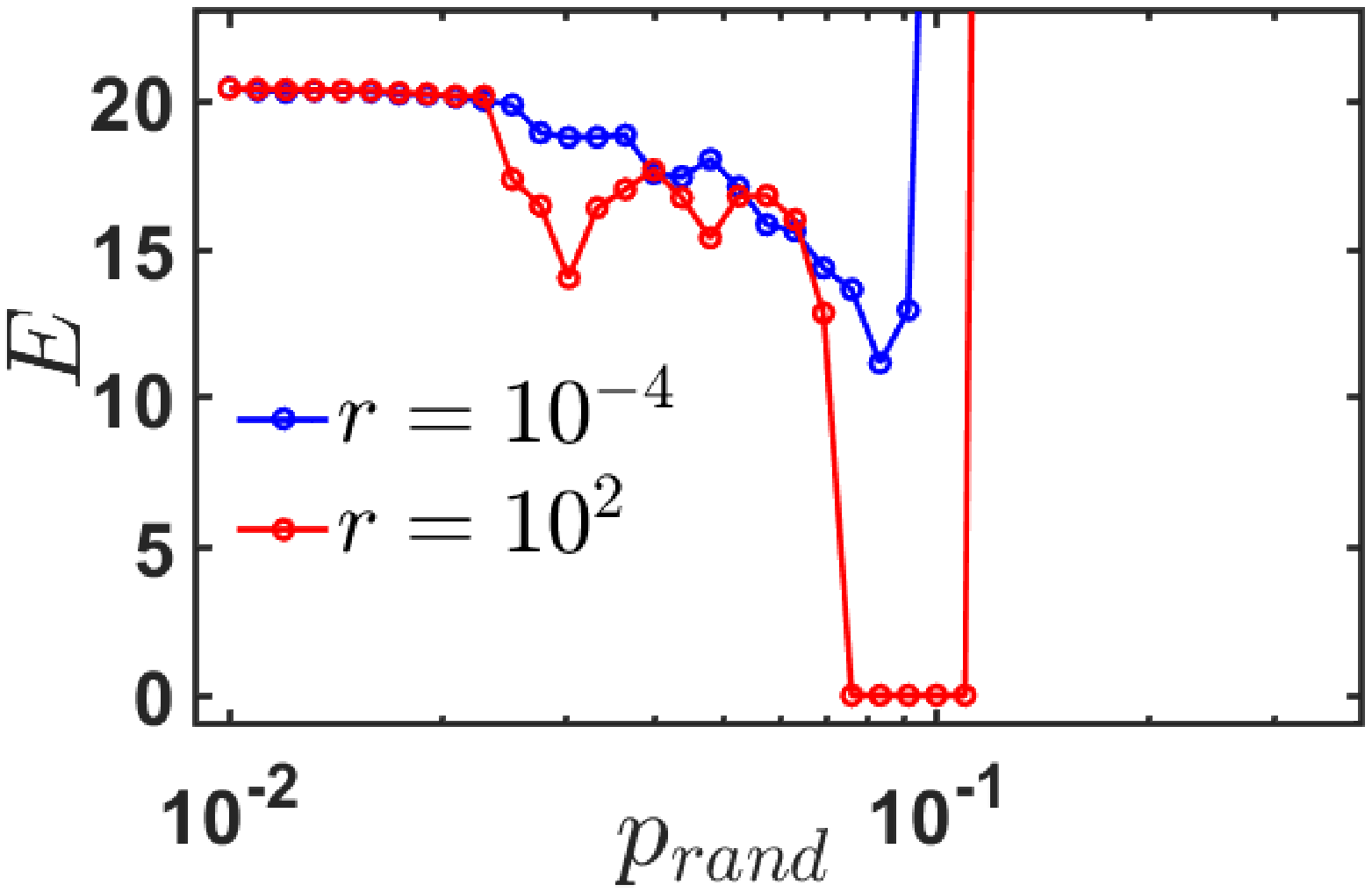}}
     	\caption{ Synchronization in temporal simplicial complex for strong higher-order coupling $\epsilon_{2}=10^{-2}$ along with fixed pairwise coupling $\epsilon_{1}=0.1$. The variation of average synchronization error $E$ as a function of connectivity probability $p_{rand}$ in almost static $(r=10^{-4})$ and sufficiently fast switching $(r=10^{2})$ regime is depicted by blue and red circled curves, respectively. }
     	\label{lorenz_prand_r_fast}
     \end{figure}                
     \subsection{Application on real-world connectivity structure}
     So far, our results have been described on synthetic network structures, i.e., the connections between the individuals are made following random network algorithm, small-world graph model, etc. But, these synthetic models are not always able to illustrate the connection mechanism between the individuals of a real-life complex system. So, we proceed with a real-world connectivity structure to apply our analytical stability theory. As for example, we consider the temporal network model describing face-to-face interactions between individuals of a French primary school registered in the {\it Sociopatterns} project \cite{stehle2011high,gemmetto2014mitigation}. In the original model, a total of $N=242$ individuals interact with each other over a period of time and connection between any two individuals is considered as active if it last for at least $20$ seconds. As a result, a set of $332$ networks are generated over the total period of time. But, at each time instant there are at least two individuals that do not have a connection path between them, i.e., none of these networks are connected at a particular time \cite{stehle2011high}. However, the union of these networks over the total period of time yields a connected network. A necessary condition for complete synchronization in a graph is its connectedness. Complete synchrony does not emerge in a disconnected network and as a result the networks generated at every $20$ seconds in this scenario does not exhibit complete synchronization, but the aggregated network does (Fig. \ref{time_aggregated} illustrates the structure of time aggregated network).
     \begin{figure*}[ht] 
     	\centerline{
     		\includegraphics[scale=0.15]{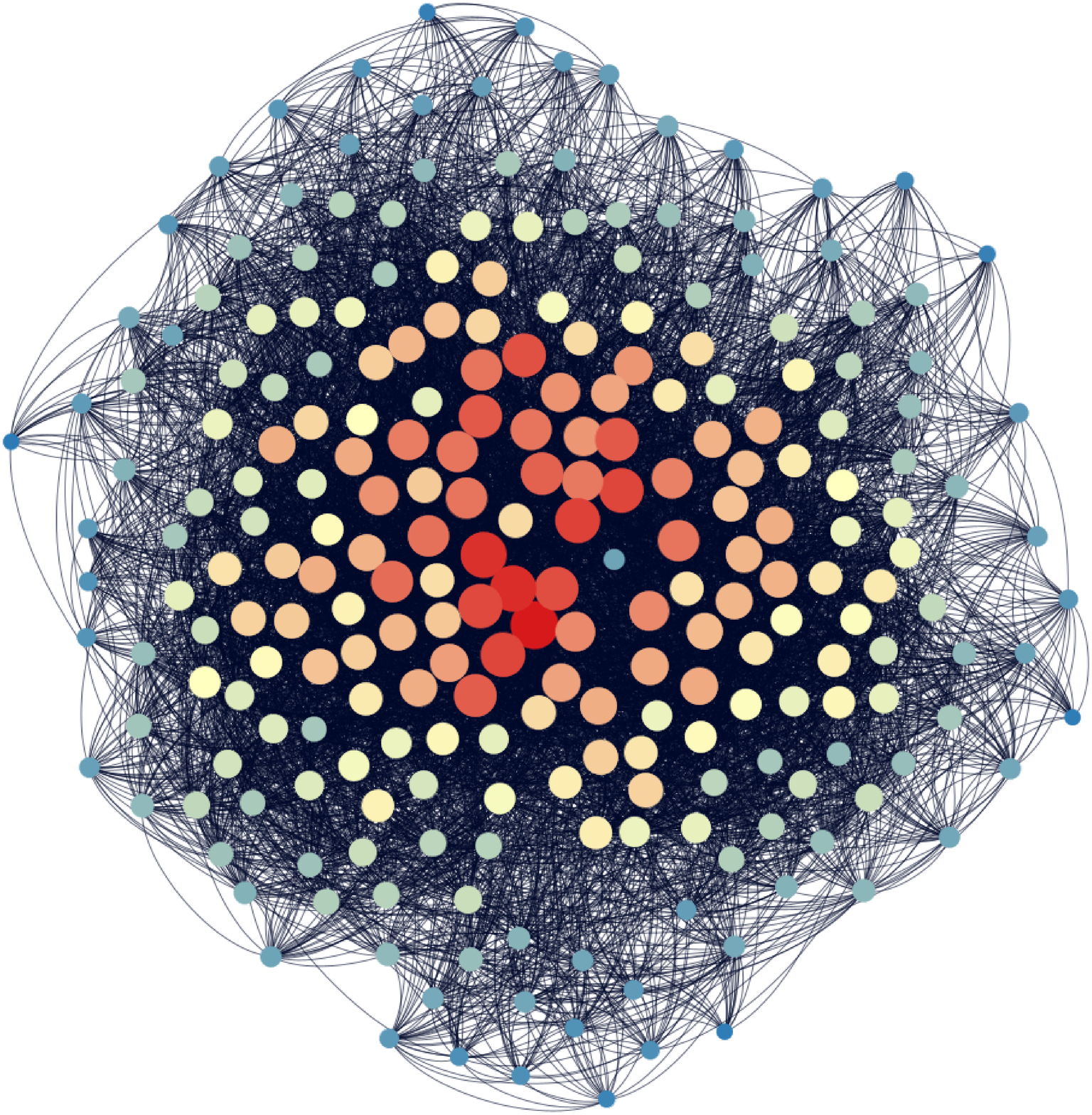}
     		\includegraphics[scale=0.4]{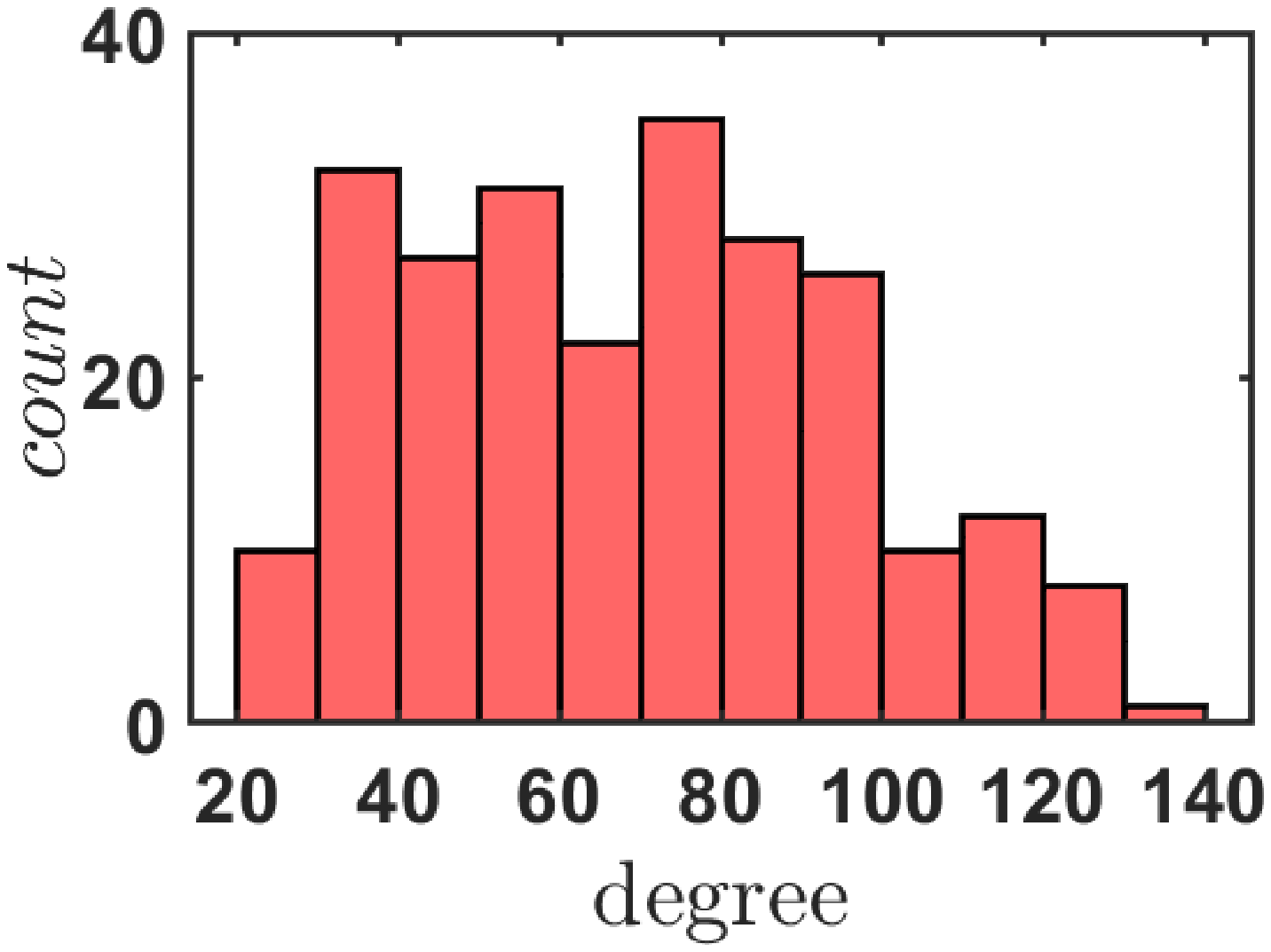}}
     	\caption{The left panel describes the time aggregated structure of the temporal face-to-face interactions between individuals of French primary school. The corresponding degree distribution is displayed in right panel which confirms the connectedness of the aggregated structure.}
     	\label{time_aggregated}
     \end{figure*}
     Now, the interactions among any two individuals in these networks are considered to be only pairwise, i.e., the interactions are considered when any two individuals interacts face-to-face. But, there may be situations where a group of students can have face-to-face interactions or a teacher can interact with a group of students at a time instant. One can assume these group interactions here by observing the occurrence of several triangles ($2$-simplices) or any other higher-dimensional simplices. However, we consider only the interactions up to three-body. In this way, the links represent face-to-face interaction among any two individuals and the triangles ($2$-simplices) indicate face-to-face interactions among three individuals at an instant of time (here, all the triangles are considered to represent three-body interactions). When mimicking the development of opinions in these complex groups, oscillators are generally utilized to explain the components of a coupled dynamical system \cite{castellano2009statistical}. Since it can be interesting to explore synchronization in a more generic scenario in which opinions do not necessarily coalesce to a fixed steady state, we have chosen to employ chaotic oscillators as dynamical units \cite{pluchino2006opinion}. Particularly, we choose chaotic R\"{o}ssler oscillator to represent the evolution of each dynamical units with the pairwise and three-body coupling schemes as $H^{(1)}({\bf x}_i,{\bf x}_{j})=[0,\; y_{j}^{3}-y_{i}^{3},\; 0]^{tr}$ and $H^{(2)}({\bf x}_i,{\bf x}_{j},{\bf x}_{k})=[0,\; y_{j}^{2}y_{k}-y_{i}^{3},\; 0]^{tr}$, respectively. Now, since for adequately fast switching, the time-averaged connectivity structure plays the role of synchronization indicator, we use the time-averaged coupling matrices to investigate the emergence of synchronization. Again, the considered coupling functions satisfy the specific coupling condition \eqref{natural}, therefore, solving the decoupled master stability equation \eqref{natural_stability_4} for maximum Lyapunov exponent $\Lambda_{max}$ provides the synchronization thresholds for the considered simplicial complex.  
     \begin{figure}[ht] 
     	\centerline{
     		\includegraphics[scale=0.35]{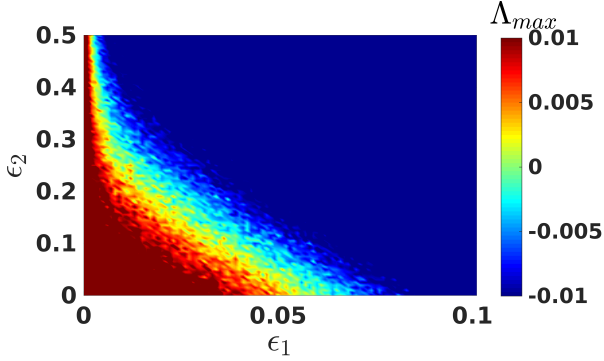}}
     	\caption{{\bf Synchronization in real-world time-varying structure.} Variation of maximum Lyapunov exponent $\Lambda_{max}$ in the $(\epsilon_{1},\epsilon_{2})$ parameter plane for time-averaged connection structure. Vertical colorbar represents the variation of $\Lambda_{max}$ with deep blue describing the region of synchronization obtained from Eq. \eqref{natural_stability_4}.}
     	\label{real_eps1_eps2}
     \end{figure}
     \par Figure \ref{real_eps1_eps2} depicts the region of synchronization by concurrently varying pairwise coupling strength $\epsilon_{1} \in [0,\;0.1]$ and triadic coupling strength $\epsilon_{2} \in [0,\;0.5]$. Parameter region with negative value of $\Lambda_{max}$ corresponds to stable synchronization domain and the unstable synchronization region is represented by the positive value. As observed, the critical value of $\epsilon_{1}$ for the emergence of complete synchronization decreases with the increasing value of triadic coupling strength $\epsilon_{2}$, i.e., the presence of triadic interactions play an important role in the enhancement of synchronization. Figure \ref{real_eps1_eps2} displays another interesting observation that the sole presence of three-body interactions (i.e., when $\epsilon_{1}=0$) is not sufficient for the achievement of synchronization, while the occurrence of synchronous state can be guaranteed when only pairwise interactions are taken into consideration.         
     \section{Discussions} \label{conclusion}
     In this article, we looked into the combined impact of higher-order interactions and temporal connectivity for the emergence of collective phenomena in complex networked systems. Particularly, we investigated the synchronization phenomenon in temporal simplicial complexes. To do this, we proposed the most general mathematical model accounting for group interactions of any order between interactive units of arbitrary nature, coupled through coupling functions of arbitrary functional form. Under the existence and invariance conditions of synchronization solution, we acquire the necessary requirement for stable synchronization state when the connectivity structure of the simplicial complex varies rapidly in time. We have shown that the time-averaged simplicial structure plays the role of synchronization indicator when the switching among connectivity structures happens quickly. Our stability approach generalizes the conventional MSF scheme to temporal higher-order structures, although the complexity of temporal connectivity and group interactions result in a high dimensional master stability equation coupled through generalized time-averaged Laplacians. But, we have provided some relevant instances for which our approach mimics the MSF conjecture and provides fully decoupled master stability equation with dimension equal to the dimension of an individual dynamical unit. All our analytical deductions have been accomplished by numerical results on synthetic and real-world connection topologies. Our investigation shows that the region of synchronization significantly increases for rapid switching. Additionally, we also demonstrate that the synchronous state is attainable even in low connectivity regime of the temporal simplicial complex with significant higher-order coupling and sufficiently quick rewiring.
     \par We point out that our theoretical analysis has two primary restrictions. The first limitation is that the stability criterion is eligible for only rapid switching among the structures of simplicial complexes. We are unable to draw a conclusion about the stability of synchronization for arbitrary switching. So, in near future the stability of synchronous state in temporal simplicial complex for arbitrary switching may be investigated. The other limitation is that we have provided only the local stability of synchronization solution which does not provide any insight about the global stability. As a result, another intriguing area for future research is the examination of the global stability of synchronous state in temporal higher-order structure.
     \par In spite of these limitations, we expect that our findings open the door for various innovative studies. Due to the proposition of a generalized approach irrespective of governing system, coupling schemes and order of higher-order interactions, our discoveries might become beneficial in a variety of practical and experimental situations to derive a number of theoretical predictions on the occurrence of synchronous state.                     
     \appendix  
     \section{Stability of synchronization in temporal simplicial complexes with generic coupling functions}\label{generic coupling}  
     \par We now extend our results to synchronization stability problem in temporal simplicial complex \eqref{gen_model} with arbitrary linear or nonlinear pairwise and higher-order coupling schemes between the unitary elements. This implies that in this scenario the coupling functions do not satisfy the synchronization non-invasive condition \eqref{noninvasive}, i.e., $ H^{(m)}({\bf x}_0,{\bf x}_0,{\bf x}_0,\cdots,{\bf x}_0) \ne 0$, for all $m=1,2,\cdots,P$. 
     As we remove restrictions from the coupling functions and consider any arbitrary type of coupling functions, the invariance condition for the complete synchrony does not remain trivial anymore.
     \par Therefore to derive the invariance condition we assume that each node of the simplicial complex begins to evolve with the synchronous state $\mathbf{x}_0$ at a time instance $t = t_0$. Then, the $i^{\mbox{th}}$ node of the simplicial complex advances according to the following equation of motion, 
     \begin{widetext}
     	\begin{equation}\label{invariance_1}
     	\begin{array}{l}
     		\dot{\bf x}_i(t_{0}) = F({\bf x}_0)+ \epsilon_{1}d_{(in)_i}^{(1)}(t_{0}) H^{(1)}({\bf x}_0,{\bf x}_0) + 2 \epsilon_{2} d_{(in)_i}^{(2)}(t) H^{(2)}({\bf x}_0,{\bf x}_0,{\bf x}_0) \\ ~~~~~~~~~~~~~~~+\cdots+P! \epsilon_{P} d_{(in)_i}^{(P)}(t) H^{(P)}({\bf x}_0,{\bf x}_0,{\bf x}_0,\cdots,{\bf x}_0),
     	\end{array}
     \end{equation}
     \end{widetext}
     where $d_{(in)_i}^{(1)}(t_{0})= {\sum\limits_{j=1}^{N}{\mathscr{B}_{ij}^{(1)}}}(t_{0})$ denotes the number of links incident to node $i$ at time $t_{0}$, i.e., in-degree of node $i$ at time $t_{0}$ and $d_{(in)_i}^{(m)}(t_{0})= \frac{1}{m!}{\sum\limits_{k_1=1}^{N}\sum\limits_{k_2=1}^{N}\cdots\sum\limits_{k_m=1}^{N}{\mathscr{B}_{ik_1k_2\cdots k_m}^{(2)}}}(t_{0})$, $m=2,3,\cdots,P$ indicates the number of $m$-simplices incident to node $i$ at the time instance $t_{0}$, i.e., generalized in-degree of node $i$ at the time instance $t_{0}$. Now, to preserve the synchronization solution, every node must advance at the same rate. Therefore, the velocity of two different nodes $i$ and $j$ of the simplicial complex should be the same, i.e.,
     \begin{equation} \label{equal_velocity}
     	\begin{array}{l}
     		\dot{\mathbf{x}}_j(t_{0}) = \dot{\mathbf{x}}_i(t_{0}), \;\; {\mbox{for}~j\ne i,~\mbox{and}}~j,i=1,2,\dots,N.	
     	\end{array}
     \end{equation}
     Since we consider generic coupling functions, so to satisfy Eq. \eqref{equal_velocity} the connection topologies of the temporal simplicial complex \eqref{gen_model} must fulfill the following condition, 
     \begin{equation} \label{invariance_2}
     	\begin{array}{l}
     		d_{(in)_i}^{(m)}(t_{0})=d_{(in)_j}^{(m)}(t_{0}), \; m=1,2,\cdots,P  .
     	\end{array}
     \end{equation}
     {\it Therefore, an equal number of $m$-simplices must incident to every node at any specific time instance in order to achieve the synchronization solution.} But, in this work we consider only the simplicial complexes with undirected pairwise and group interactions. So, in the present instance, $d_{(in)}^{(1)}(t)=d^{(1)}(t)$, the degree of each node and $d_{(in)}^{(m)}(t)=d^{(m)}(t)$, the generalized degree of each node must be equal at any specific time instance. We further consider that the number of edges and $m$-simplices incident to the nodes are time-independent. 
     \par As a result, the synchronization solution $\bf x_{0}$ advances as follows,
     \begin{equation}\label{sync_sol2}
     	\begin{array}{l}
     		\dot{\bf x}_0(t_{0}) = F({\bf x}_0)+\epsilon_{1} d^{(1)} H^{(1)}({\bf x}_0,{\bf x}_0)  + 2 \epsilon_{2} d^{(2)} H^{(2)}({\bf x}_0,{\bf x}_0,{\bf x}_0) \\ ~~~~~~~~~~~~~ +\cdots + \epsilon_{P} P! d^{(P)} H^{(P)}({\bf x}_0,{\bf x}_0,{\bf x}_0,\cdots,{\bf x}_0).
     	\end{array}
     \end{equation}
     \par Now to study the stability of synchronization solution ${\bf x}_{0}$, we introduce a small perturbation around the synchronization manifold as $\delta {\bf x}_i= {\bf x}_i-{\bf x}_0$ and perform the linear stability analysis. Hence, the variational equation in terms of the stake variable $\delta {\bf x}_i$ becomes follows
     \begin{widetext}
     	\begin{equation}\label{nonlinear_time_var1}
     	\begin{array}{l}
     		\delta \dot{\bf x}_{i}= \big[Jf({\bf x}_{0})+\sum_{m=1}^{P}\epsilon_{m}m!d^{(m)}\mathcal{D}H^{(m)}\big]\delta {\bf x}_{i}- \sum_{m=1}^{P}\epsilon_{m} \sum_{j=1}^{N} \mathscr{L}^{(m)}_{ij}(t) \mathcal{J}H^{(m)}\delta {\bf x}_{j},
     	\end{array}
     \end{equation}
     \end{widetext} 
     where for each $m=1,2,\cdots,P$, $\mathcal{D}H^{(m)}$ represents the sum of all partial derivatives of the coupling function $H^{(m)}({\bf x}_i,{\bf x}_{k_1},{\bf x}_{k_2},\cdots,{\bf x}_{k_m})$ at the synchronization solution, i.e., $\mathcal{D}H^{(m)}=H^{(m)}_{{\bf x}_{i}}({\bf x}_{0},{\bf x}_{0},\cdots,{\bf x}_{0})+H^{(m)}_{{\bf x}_{k_1}}({\bf x}_{0},{\bf x}_{0},\cdots,{\bf x}_{0})+\cdots+H^{(m)}_{{\bf x}_{k_m}}({\bf x}_{0},{\bf x}_{0},{\bf x}_{0},\cdots,{\bf x}_{0})$.
     \par Following the conceptual steps discussed in Sec. \ref{stability analysis}, one can easily prove that in this case also the time-averaged simplicial structure plays the role of synchronization indicator. Consequently, the dynamics of the transverse modes associated with the time-averaged structure provide the condition for synchronization stability. Proceeding with the similar analysis presented in Sec. \ref{stability analysis}, we can derive that in the present scenario, the dynamics of the transverse modes are given by the following set of coupled linear equations,
     \begin{widetext}
     	\begin{equation}\label{nonlinear_average_trans2}
     	\begin{array}{l}
     		\dot{\bar{\zeta}}_{{\perp}_{i}}= [Jf({\bf x}_{0})+\sum_{m=1}^{P}\epsilon_{m}m!d^{(m)}\mathcal{D}H^{(m)}-\epsilon_{1}{\lambda}^{(1)}_{i} \mathcal{J}H^{(1)}]{\bar{\zeta}}_{{\perp}_{i}} \\ ~~~~~~~~~~~~~~ -\sum_{m=2}^{P}\epsilon_{m}\sum_{j=2}^{N}\bar{W}^{(m)}_{{2}_{ij}} \mathcal{J}H^{(m)} \bar{\zeta}_{\perp_{j}}, \;\;\;\; i=2,3,\cdots,N.
     	\end{array}
     \end{equation}
     \end{widetext}
     This is the required master stability equation for complete synchronization in a temporal simplicial complex with arbitrary linear or non-linear coupling schemes. Solving equation \eqref{nonlinear_average_trans2} for the calculation of maximum Lyapunov exponent $\Lambda_{max}$, and finding out the region of negative $\Lambda_{max}$ as a function of coupling strengths and connectivity topology gives the sufficient condition for synchronization stability.
     \begin{remark}
     	We would like to emphasize on the fact that all the relevant instances regarding dimension reduction of the master stability equation discussed earlier in sections \ref{natural_coupling} and \ref{reduction_2} also works perfectly in the present case, only under an extra condition \eqref{invariance_2} which is however necessary for the existence of synchronous solution.   
     \end{remark}

	\section*{Data availability}
     The data corresponding to the real-world example is publicly available at \url{http://www.sociopatterns.org}.

	\bibliographystyle{apsrev4-1} % Tell bibtex which bibliography style to use
	\bibliography{temporal_simplicial_references}
	
\end{document}